%% file: main.tex
\documentclass{article}
\usepackage{fullpage}
\usepackage[authoryear,sort]{natbib}

\usepackage{extarrows}

\usepackage{wrapfig}
\usepackage{amsmath,amssymb,mathtools}
\usepackage{bm}
\usepackage{amsthm}
\usepackage{graphicx,subfig,color,wrapfig,epsfig,epstopdf}
\usepackage{enumitem}
\usepackage{booktabs}
\usepackage{algorithm,algorithmic}

\renewcommand{\paragraph}[1]{\noindent\textbf{#1}}

\usepackage[utf8]{inputenc} 
\usepackage[T1]{fontenc}    
\usepackage{hyperref}       
\usepackage{url}            
\usepackage{booktabs}       
\usepackage{amsfonts}       
\usepackage{nicefrac}       
\usepackage{microtype}      

\usepackage{color}

\def\bc{\color{black}}

\usepackage[sc]{mathpazo} 
\usepackage[T1]{fontenc}  
\usepackage{microtype}    
\usepackage{abstract}     
\newcommand\numberthis{\addtocounter{equation}{1}\tag{\theequation}}

\allowdisplaybreaks[4]

\newcounter{ass_counter}
\newcounter{thm_counter}

\newtheorem{theorem}[thm_counter]{Theorem}
\newtheorem{lemma}[thm_counter]{Lemma}
\newtheorem{corollary}[thm_counter]{Corollary}
\newtheorem{assumption}[ass_counter]{Assumption}

\usepackage{amsmath,amssymb,mathtools}
\usepackage{bm}
\usepackage{amsthm}
\usepackage{graphicx,subfig,color,wrapfig,epsfig,epstopdf}
\usepackage{enumitem}
\usepackage{booktabs}
\usepackage{algorithm,algorithmic}

\DeclareMathOperator{\Tr}{Tr}

\title{Distributed Learning over Unreliable Networks}
\date{}

\usepackage{authblk}
\author[1]{Chen Yu\thanks{\texttt{cyu28@ur.rochester.edu}}}
\author[1]{Hanlin Tang\thanks{\texttt{htang14@ur.rochester.edu}}}
\author[2]{Cedric Renggli\thanks{\texttt{cedric.renggli@inf.ethz.ch}}}
\author[2]{Simon Kassing\thanks{\texttt{simon.kassing@inf.ethz.ch}}}
\author[2]{Ankit Singla\thanks{\texttt{ankit.singla@inf.ethz.ch}}}
\author[4]{Dan Alistarh\thanks{\texttt{d.alistarh@gmail.com}}}
\author[2]{Ce Zhang\thanks{\texttt{ce.zhang@inf.ethz.ch}}}
\author[1,3]{Ji Liu\thanks{\texttt{ji.liu.uwisc@gmail.com}}}
\affil[1]{Department of Computer Science, University of Rochester}
\affil[2]{Department of Computer Science, ETH Zurich}
\affil[3]{Seattle AI Lab, FeDA Lab, Kwai Inc}
\affil[4]{Institute of Science and Technology Austria}

\begin{document}

\maketitle

\begin{abstract}
Most of today's distributed machine learning systems
assume {\em reliable networks}: whenever
two machines exchange information (e.g., gradients
or models), the network should guarantee the delivery
of the message. At the same time, recent 
work exhibits 
the impressive tolerance of machine learning algorithms to errors or noise arising from relaxed communication or synchronization. 
In this paper, we connect these two trends, and consider the  
following question: {\em Can we design machine
learning systems that are tolerant to network unreliability during training?}
With this motivation, we focus on a theoretical
problem of independent interest---given a standard distributed parameter server 
architecture, if every communication between
the worker and the server has a non-zero probability $p$
of being dropped, does there exist 
an algorithm that still converges,
and at what speed?
The technical contribution of this paper is a novel theoretical analysis proving that 
 distributed learning over unreliable network can achieve comparable convergence rate to centralized or distributed learning over reliable networks. 
 Further, we prove that the influence of the packet drop rate diminishes with the growth of the number of \textcolor{black}{parameter servers}.
We map this theoretical result onto a real-world scenario, training deep neural networks over
an unreliable network layer, and conduct network simulation
to validate the system improvement by allowing the networks to be unreliable.
\end{abstract}

\section{Introduction}

Distributed learning has attracted significant interest from both
academia and industry. Over the last decade, researchers have 
come with up a range of different designs of more efficient 
learning systems. An important subset of this work focuses on 
understanding the impact of different system relaxations to
the convergence and performance of distributed stochastic gradient
descent, such as the compression of communication, e.g~\citep{seide2016cntk},
decentralized communication~\citep{lian2017can,sirb2016consensus,lan2017communication,Tang:2018aa,Stich:2018aa}, and asynchronous
communication~\citep{lian2017asynchronous,zhang2013asynchronous,lian2015asynchronous}. 
Most of these works are motivated by real-world system bottlenecks, abstracted as general problems for the purposes of analysis.

In this paper, we focus on the centralized SGD algorithm in the distributed machine learning scenario implemented using standard AllReduce operator or parameter server ar- chitecture and we are motivated by a new type of system relaxation—the reliability of the communication channel. We abstract this problem as a theoretical one, conduct a novel convergence analysis for this scenario, and then vali- date our results in a practical setting.


\begin{figure} 
\centering
\includegraphics[width=0.5\textwidth]{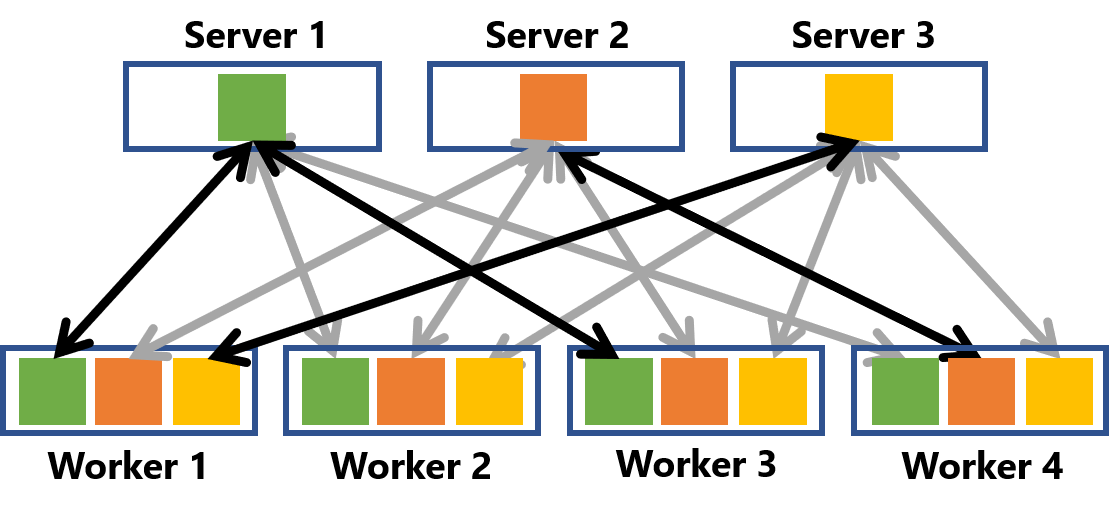}
\caption{An illustration of the communication pattern 
of distributed learning with three parameter servers
and four workers --- each server serves a partition
of the model, and each worker holds a replica of
the whole model. In this paper, we focus on the
case in which every communication between
the worker and the server has a non-zero probability $p$
of being dropped.}
\label{fig: allreduce}
\end{figure}

The \textit{Centralized SGD} algorithm works as Figure \ref{fig: allreduce} shows. Given $n$ machines/workers, each maintaining its own local model, each machine alternates local SGD steps with global communication steps, in which machines exchange their {\em local models}.
In this paper, we covers two standard distributed settings: the Parameter Server model\footnote{Our modeling above fits the case of $n$ workers and $n$ parameter servers, although our analysis will extend to any setting of these parameters.}~\cite{li2014scaling,abadi2016tensorflow}, as well as standard implementations of the AllReduce averaging operation in a decentralized setting~\cite{seide2016cntk,renggli2018sparcml}. There are two components in the
communication step:

\begin{enumerate}
\item {\bf Step 0 - Model Partitioning (Only Conducted Once).}
In most state-of-the-art implementations
of AllReduce and parameter servers~\cite{li2014scaling,abadi2016tensorflow,Thakur:2005:OCC:2747766.2747771},
models are partitioned into $n$ blocks,
and each machine is the owner of one block~\cite{Thakur:2005:OCC:2747766.2747771}.
The rationale is to increase parallelism 
over the utilization of the underlying 
communication network. The partitioning
strategy does not change during training.

\item {\bf Step 1.1 - Reduce-Scatter.}
In the Reduce-Scatter step, 
for each block (model partition) $i$, the machines average their model on the block by sending it to the corresponding machine.  

\item {\bf Step 1.2 - All-Gather.}
In the subsequent All-Gather step, each machine broadcasts its block to all others, so that all machines have a consistent model copy. 

\end{enumerate}

In this paper, we focus on the
scenario in which the communication is unreliable --- The communication channel between any two machines has a probability $p$ of not delivering a message, as the Figure \ref{fig: allreduce} shows, where the grey arrows represent the dropping message and the black arrows represent the success-delivering message. 
We change the two aggregation steps as follows. In the Reduce-Scatter step, a {\em uniform random} subset of machines will average their model on each model block $i$. 
In the All-Gather step, it is again a {\em uniform random} subset of machines which receive the resulting average. 
Specifically, machines not chosen for the Reduce-Scatter step do not contribute to the average, and all machines that are not chosen  for the All-Gather will not receive updates on their model block $i$.
This is a realistic model of running an AllReduce operator implemented with Reduce-Scatter/All-Gather on unreliable network. We call this revised algorithm the RPS algorithm.

Our main technical contribution is characterizing the convergence properties of the RPS algorithm. To the best of our knowledge, this is a novel theoretical analysis of this faulty communication model.
We will survey related work in more
details in Section~\ref{sec:related-work}.


We then apply our theoretical result to a real-world use case,
illustrating the potential benefit of allowing an unreliable network. We focus on a realistic scenario where the network is shared among multiple applications or tenants, for instance in a data center. 
Both applications communicate using the same network. 
In this case, if the machine learning traffic is tolerant to some packet loss, the other application can potentially be made faster by receiving priority for its network traffic. 
Via network simulations, we find that tolerating a $10\%$ drop rate for the learning traffic can make a simple (emulated) Web service up to $1.2\times$ faster
(Even small speedups of $10\%$ are significant for such services; for example, Google actively pursues minimizing its Web services' response latency). 
At the same time, this degree of loss does not impact the convergence rate for a range of machine learning applications, such as image classification and natural language processing.

\paragraph{Organization}
The rest of this paper is organized as follow. We first review some related work in Section \ref{sec:related-work}. Then we formulate the problem in Section \ref{sec:problem} and describe the RPS algorithm in Section \ref{sec:algorithm}, with its theoretical guarantee stated in Section \ref{sec:theorem}. We evaluate the scalability and accuracy of the RPS algorithm in Section \ref{sec:experiments} and study an interesting case of speeding up colocated applications in Section \ref{sec:case}. At last, we conclude the paper in
Section \ref{sec:conclusion}. The proofs of our theoretical results can be found in
the supplementary material.
\section{Related Work} \label{sec:related-work}
{\bc
\paragraph{Distributed Learning} 
There has been a huge number of works on distributing deep learning, e.g.~\citet{seide2016cntk,abadi2016tensorflow,goyal2017accurate,colin2016gossip}. Also, many optimization algorithms are proved to achieve much better performance with more workers. For example,  \citet{hajinezhad2016nestt} utilize a primal-dual based method for optimizing a finite-sum objective function and proved that it's possible to achieve a $\mathcal{O}(n)$ speedup corresponding to the number of the workers. In \citet{xu2017adaptive}, an adaptive consensus ADMM is proposed and \citet{goldstein2016unwrapping} studied the performance of transpose ADMM on an entire distributed dataset. In \citet{scaman2017optimal}, the optimal convergence rate for both centralized and decentralized distributed learning is given with the time cost for communication included. In \citet{Lin:2018aa,local_sgd}, they investigate the trade off between getting more mini-batches or having more communication.  To save the communication cost, some sparse based distributed learning algorithms is proposed \citep{shen2018towards,wu2018error,wen2017terngrad,mcmahan2016communication,wang2016efficient}. Recent works indicate that many distributed learning is delay-tolerant under an asynchronous setting \citep{zhou2018distributed,lian2015asynchronous,sra2015adadelay,leblond2016asaga}. Also, in \citet{blanchard2017machine,Yin2018ByzantineRobustDL,Alistarh2018ByzantineSG} They study the Byzantine-robust distributed learning when the Byzantine worker is included in the network. In \citet{NIPS2018_7327}, authors proposed a compressed DNN training strategy in order to save the computational cost of floating point.

}



\paragraph{Centralized parallel training} 
Centralized parallel \citep{agarwal2011distributed, recht2011hogwild} training works on the network that is designed to ensure that all workers could get information of all others.
One communication primitive
in centralized training is
to average/aggregate 
all models, which is
called {\em a collective 
communication operator}
in HPC literature~\cite{Thakur:2005:OCC:2747766.2747771}. Modern
machine learning systems
rely on different implementations,
e.g., parameter server model~\cite{li2014scaling,abadi2016tensorflow}
or the standard implementations of the AllReduce averaging operation in a decentralized setting  \citep{seide2016cntk,renggli2018sparcml}.
In this work, we
focus on the behavior of
centralized ML
systems under unreliable network,
when this primitive is implemented
as a distributed 
parameter servers~\cite{Jiang:2017:HDP:3035918.3035933}, which
is similar to a Reduce-Scatter/All-Gather
communication paradigm.
For many implementations of 
collective 
communication operators,
partitioning the model is one key
design point to reach the peak
communication performance~\cite{Thakur:2005:OCC:2747766.2747771}.

\paragraph{Decentralized parallel training}
Another direction of related work considers decentralized learning. Decentralized learning algorithms can be divided into \emph{fixed} topology algorithms and \emph{random} topology algorithms. There are many work related to the \emph{fixed} topology decentralized learning. 
Specifically,~\citet{jin2016scale} proposes to scale the gradient aggregation process via a gossip-like mechanism. \citet{lian2017can} provided strong convergence bounds for a similar algorithm to the one we are considering, in a setting where the communication graph is fixed and regular. In \citet{tang2018d}, a new approach that admits a better performance than decentralized SGD when the data among workers is very different is studied. \citet{pmlr-v80-shen18a} generalize the decentralized optimization problem to a monotone operator. In \citet{He:2018aa}, authors study a decentralized gradient descent  based algorithm (\textbf{CoLA}) for learning of linear classification and regression model.
For the \emph{random} topology decentralized learning, the weighted matrix for randomized algorithms can be time-varying, which means workers are allowed to change the communication network based on the availability of the network. There are many works \citep{boyd2006randomized, li2010consensus,lobel2011distributed, nedic2017achieving,nedic2015distributed} studying the random topology decentralized SGD algorithms under different assumptions. \citet{blot2016gossip} considers a more radical approach, called GoSGD, where each worker exchanges gradients with a random subset of other workers in each round. They show that GoSGD can be faster than Elastic Averaging SGD~\citep{zhang2015deep} on CIFAR-10, but provide no large-scale experiments or theoretical justification. 
Recently,~\citet{daily2018gossipgrad} proposed GossipGrad, a more complex gossip-based scheme with upper bounds on the time for workers to communicate indirectly, periodic rotation of partners and shuffling of the input data, which provides strong empirical results on large-scale deployments. 
The authors also provide an informal justification for why GossipGrad should converge.

In this paper, we consider a general model communication, which covers both Parameter Server~\citep{li2014scaling} and AllReduce~\citep{seide2016cntk} distribution strategies. 
We specifically include the uncertainty of the network into our theoretical analysis. 
In addition, our analysis highlights the fact that the system can handle additional packet drops as we increase the number of worker nodes. 

\section{Problem Setup} \label{sec:problem}

We consider the following distributed optimization problem:
\begin{equation}
\min_{\bm{x}}\quad f(\bm{x}) = {1\over n} \sum_{i=1}^n \underbrace{\mathbb{E}_{\bm{\bm{\xi}}\sim\mathcal{D}_i}F_{i}(\bm{x}; \bm{\bm{\xi}})}_{=: f_i(\bm{x})},\label{eq:main}
\end{equation}
where $n$ is the number of workers, $D_i$ is the local data distribution for worker $i$ (in other words, we do not assume that all nodes can access the same data set), and $F_i(\bm{x}; \bm{\bm{\xi}} )$ is the local loss function of model $\bm{x}$ given data $\bm{\bm{\xi}}$ for worker $i$.

\paragraph{Unreliable Network Connection} Nodes can communicate with all other workers, but with packet drop rate $p$ (here we do not use the common-used phrase ``packet loss rate'' because we use ``loss'' to refer to the loss function). That means, whenever any node forwards models or data to any other model, the destination worker {\em fails} to receive it, with probability $p$. 
For simplicity, we assume that all packet drop events are independent, and that they occur with the same probability $p$.

\paragraph{Definitions and notations}
Throughout, we use the following notation and definitions:
\begin{itemize}
\item $\nabla f(\cdot)$ denotes the gradient of the function $f$.
\item $\lambda_{i}(\cdot)$ denotes the $i$th largest eigenvalue of a matrix.
\item $\bm{1}_n=[1,1,\cdots,1]^{\top}\in\mathbb{R}^n$ denotes the full-one vector.
\item $A_n:=\frac{\bm{1}_n\bm{1}_n^{\top}}{n}$ denotes the all $\frac{1}{n}$'s $n$ by $n$ matrix.
\item $\|\cdot\|$ denotes the  $\ell_2$ norm for vectors.
\item $\|\cdot\|_F$ denotes the Frobenius norm of matrices.
\end{itemize}

\section{Algorithm} \label{sec:algorithm}

In this section, we describe our algorithm,
namely RPS --- {\em Reliable
Parameter Server} --- as it is robust to
package loss in the network layer. 
We first describe our
algorithm in detail, followed by its interpretation from a global view. 

\subsection{Our Algorithm: RPS}
In the RPS algorithm, each worker maintains an individual local model. We use $\bm{x}^{(i)}_t$ to denote the local model on worker $i$ at time step $t$. 
At each iteration $t$, each worker first performs a regular SGD step
\begin{align*}
\bm{v}^{(i)}_t \leftarrow \bm{x}^{(i)}_t - \gamma \nabla F_i(\bm{x}^{(i)}_t; {\bm{\xi}}^{(i)}_t);
\end{align*}
where $\gamma$ is the learning rate and $\bm{\xi}_{t}^{(i)}$ are the data samples of worker $i$ at iteration $t$.

We would like to reliably average the vector $\bm{v}_t^{(i)}$ among all workers, via the RPS procedure. In brief, the RS step perfors communication-efficient model averaging, and the AG step performs communication-efficient model sharing.

\paragraph{The Reduce-Scatter (RS) step:}  In this step, each worker $i$ divides $\bm{v}_t^{(i)}$ into $n$ equally-sized blocks.
\begin{equation}\label{eq: dividev}
\bm{v}_t^{(i)} = \left(\left(\bm{v}_t^{(i,1)}\right)^{\top}, \left(\bm{v}_t^{(i,2)}\right)^{\top}, \cdots, \left(\bm{v}_t^{(i,n)}\right)^{\top}\right)^{\top}.
\end{equation}

The reason for this division is to reduce the communication cost and parallelize model averaging since we only assign each worker for averaging one of those blocks. For example, worker $1$ can be assigned for averaging the first block while worker $2$ might be assigned to deal with the third block. For simplicity, we would proceed our discussion in the case that worker $i$ is assigned for averaging the $i$th block.
 
After the division, each worker sends its $i$th block to worker $i$. Once receiving those blocks, each worker would average all the blocks it receives. As noted, some packets might be dropped. In this case,  worker $i$ averages all those blocks using
\begin{align*}
\tilde{\bm{v}}_t^{(i)} = \frac{1}{|\mathcal{N}_t^{(i)}|}\sum_{j\in\mathcal{N}_t^{(i)}}\bm{v}_t^{(i,j)},
\end{align*}
where $\mathcal{N}_t^{(i)}$ is the set of the packages worker $i$ receives (may including the worker $i$'s own package).

\paragraph{The AllGather (AG) step:} After computing $\tilde{\bm{v}}_t^{(i)}$, each worker $i$ attempts to broadcast $\tilde{\bm{v}}_t^{(i)}$ to all other workers, using the averaged blocks to recover the averaged original vector $\bm{v}_t^{(i)}$ by concatenation:
\begin{align*}
\bm{x}_{t+1}^{(i)} = \left( \left(\tilde{\bm{v}}_{t}^{(i,1)}\right)^{\top},\left(\tilde{\bm{v}}_{t}^{(i,2)}\right)^{\top},\cdots,\left(\tilde{\bm{v}}_{t}^{(i,n)}\right)^{\top}\right)^{\top}.
\end{align*}
Note that it is entirely possible that some workers in the network may not be able to receive some of the averaged blocks. 
In this case, they just use the original block. Formally, 
\begin{equation} \label{eq: dividex}
	\bm{x}_{t+1}^{(i)} = \left(\left(\bm{x}_{t+1}^{(i,1)}\right)^{\top}, \left(\bm{x}_{t+1}^{(i,2)}\right)^{\top},\cdots, \left(\bm{x}_{t+1}^{(i,n)}\right)^{\top}\right)^{\top},
\end{equation}
where
\begin{displaymath}
	\bm{x}_{t+1}^{(i,j)} = \left\{ \begin{array}{ll}
		\tilde{\bm{v}}_t^{(j)} & j\in \widetilde{\mathcal{N}}_t^{(i)}\\
		\bm{v}_t^{(i,j)} & j\notin \widetilde{\mathcal{N}}_t^{(i)}
	\end{array} \right.
\end{displaymath}
We can see that each worker just replace the corresponding blocks of $v_t^{(i)}$ using received averaged blocks.
 The complete algorithm is summarized in Algorithm~\ref{alg1}.

\begin{algorithm}[t!]
\caption{RPS}\label{alg1}
\begin{minipage}{1.0\linewidth}

\begin{algorithmic}[1]
\STATE {\bfseries Input:} Initialize all $\bm{x}^{(i)}_1, \forall i\in[n]$ with the same value, learning rate $\gamma$, and number of total iterations $T$.
\FOR{$t = 1,2,\cdots,T$}
\STATE Randomly sample $\bm{\xi}^{(i)}_t$ from local data of the $i$th worker, $\forall i\in[n]$.
\STATE Compute a local stochastic gradient based on $\bm{\xi}^{(i)}_t$ and current optimization variable $\bm{x}^{(i)}_t:\nabla F_i(\bm{x}^{(i)}_t;\bm{\xi}^{(i)}_t), \forall i\in[n]$
\STATE Compute the intermediate model $\bm{v}^{(i)}_t$ according to
\[
\bm{v}_{t}^{(i)}\gets \bm{x}_{t}^{(i)}-\gamma\nabla F_i(\bm{x}^{(i)}_t;\bm{\xi}^{(i)}_t),
\]
and divide $\bm{v}^{(i)}_t$ into $n$ blocks $\left(\left(\bm{v}_t^{(i,1)}\right)^{\top}, \left(\bm{v}_t^{(i,2)}\right)^{\top}, \cdots, \left(\bm{v}_t^{(i,n)}\right)^{\top}\right)^{\top}$.
\STATE For any $i\in [n]$, randomly choose one worker $b^{(i)}_{t}$ \footnote{Here $b_t^{(i)} \in \{1,2,\cdots ,n\}$ indicates which worker is assigned for averaging the $i$th block.} without replacement. Then, every worker attempts to send their $i$th block of their intermediate model to worker $b_t^{(i)}$. Then each worker averages all received blocks using
\begin{equation*}
	\tilde{\bm{v}}_t^{(i)} = \frac{1}{|\mathcal{N}_t^{(i)}|}\sum\limits_{j\in \mathcal{N}_t^{(j)}} \bm{v}_t^{(i,j)}.
\end{equation*}
\STATE Worker $b_t^{(i)}$ broadcast $\tilde{\bm{v}}_t^{(i)}$ to all workers (maybe dropped due to packet drop), $\forall i\in [n]$.
\STATE $\bm{x}_{t+1}^{(i)} = \left(\left(\bm{x}_{t+1}^{(i,1)}\right)^{\top}, \left(\bm{x}_{t+1}^{(i,2)}\right)^{\top},\cdots, \left(\bm{x}_{t+1}^{(i,n)}\right)^{\top}\right)^{\top}$, where
\begin{displaymath}
	\bm{x}_{t+1}^{(i,j)} = \left\{ \begin{array}{ll}
		\tilde{\bm{v}}_t^{(j)} & j\in \widetilde{\mathcal{N}}_t^{(i)}\\
		\bm{v}_t^{(i,j)} & j\notin \widetilde{\mathcal{N}}_t^{(i)}
	\end{array} \right. ,
\end{displaymath}
for all $i\in [n]$.
\ENDFOR
\STATE {\bfseries Output:} $\bm{x}^{(i)}_T$
\end{algorithmic}
\end{minipage}
\end{algorithm}

\subsection{RPS From a Global Viewpoint} 
It can be seen that at time step $t$, the $j$th block of worker $i$'s local model, that is, $\bm{x}_{t}^{(i,j)}$, is a linear combination of $j$th block of all workers' intermediate model $\bm{v}_t^{(k,j)} (k\in [n])$,
\begin{align} \label{eq: updatingrule}
X_{t+1}^{(j)} = V_t^{(j)} W_{t}^{(j)},
\end{align}
where 
\begin{align*}
X_{t+1}^{(j)}:= &\big(\bm{x}_{t+1}^{(1,j)}, \bm{x}_{t+1}^{(2,j)}, \cdots, \bm{x}_{t+1}^{(n,j)}\big)\\
V_t^{(j)} := &\big(\bm{v}_t^{(1,j)}, \bm{v}_t^{(2,j)}, \cdots, \bm{v}_t^{(n,j)}\big)
\end{align*}
and $W_t^{(j)}$ is the coefficient matrix indicating the communication outcome at time step $t$. The $(m,k)$th element of $W_t^{(j)}$ is denoted by $\left[W_t^{(j)}\right]_{m,k}$. $\left[W_t^{(j)}\right]_{m,k} \neq 0$ means that worker $k$ receives worker $m$'s individual $j$th block (that is, $v^{(m,j)}_t$), whereas $\left[W_t^{(j)}\right]_{m,k} = 0$ means that the package might be dropped either in RS step (worker $m$ fails to send) or AG step (worker $k$ fails to receive).
So $W_t^{(j)}$ is time-varying because of the randomness of the package drop. Also $W_t^{(j)}$ is not doubly-stochastic (in general) because the package drop is independent between RS step and AG step.
\begin{figure}[h]
\centering
\begin{minipage}{.4\textwidth}
  \centering
  \includegraphics[width=1.0\linewidth]{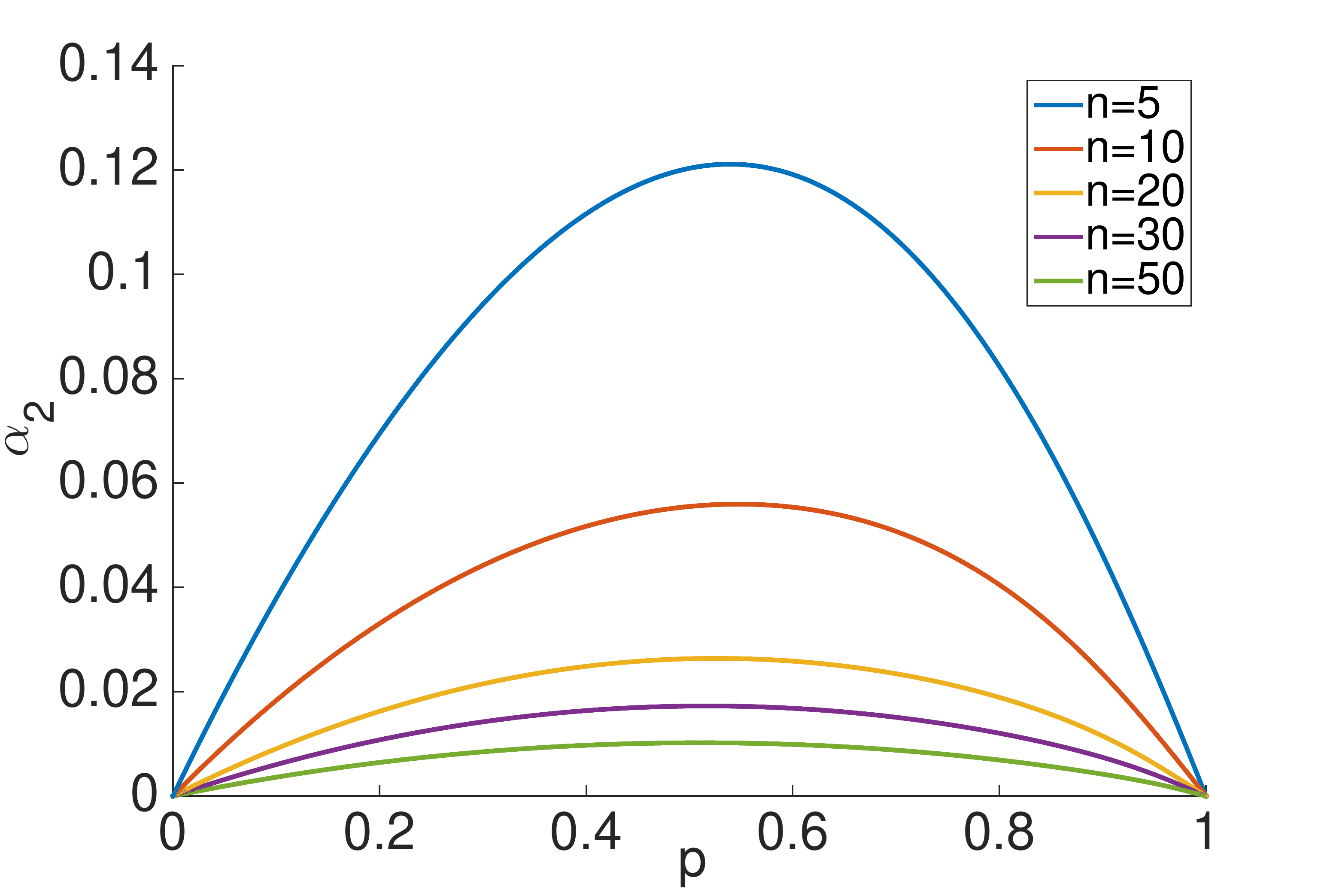}
  \captionof{figure}{$\alpha_2$ under different number of workers n and package drop rate $p$.}
  \label{fig_alpha2}
\end{minipage}%
\hspace{0.03\textwidth}%
\begin{minipage}{.4\textwidth}
  \centering
  \includegraphics[width=1.0\linewidth]{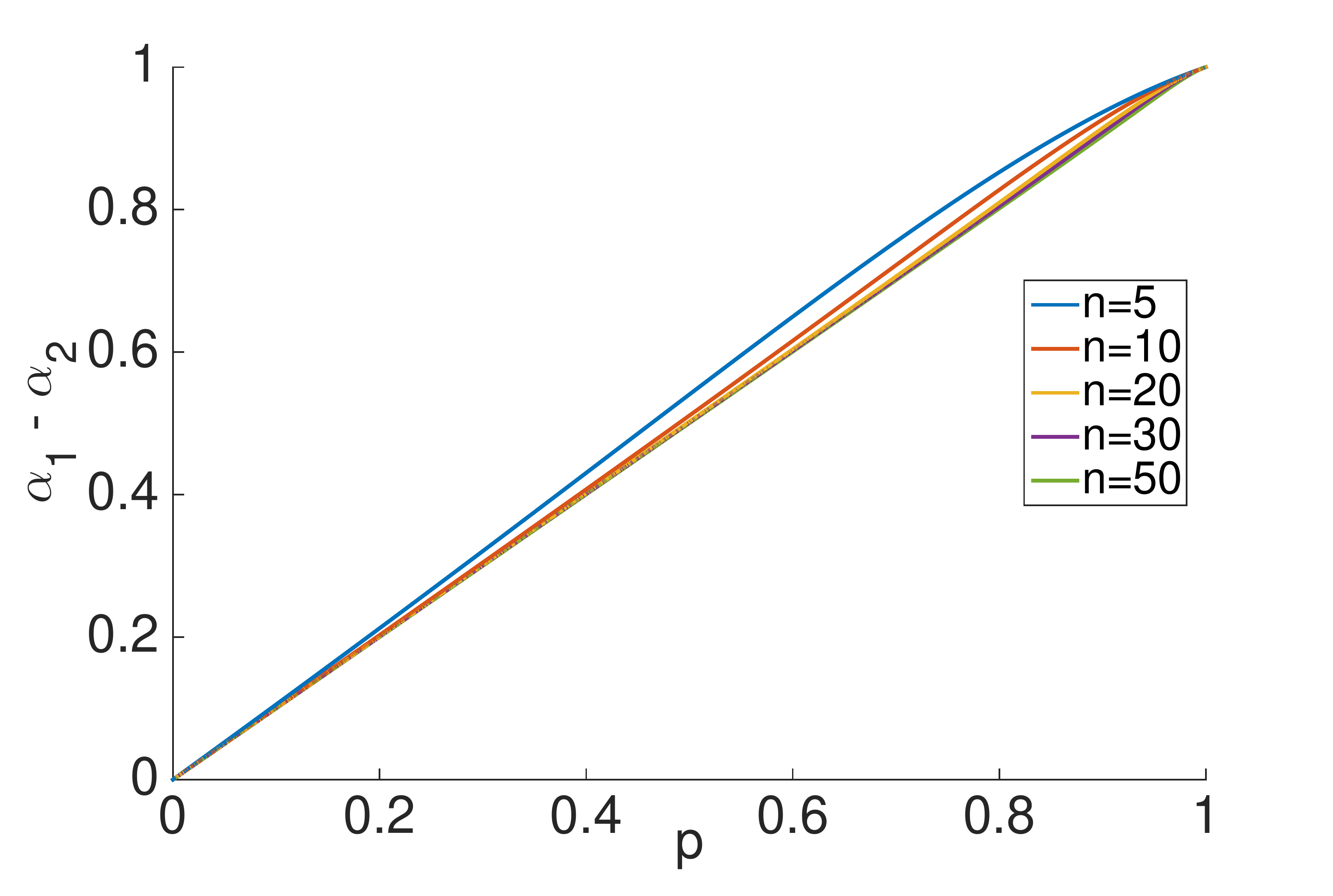}
  \captionof{figure}{$(\alpha_1 - \alpha_2)$ under different number of workers n and package drop rate $p$.}
  \label{fig_beta}
\end{minipage}%
\end{figure}

\paragraph {The property of $W^{(j)}_t$}
In fact, it can be shown that all $W^{(j)}_t$'s ($\forall j, \forall t$) satisfy the following properties 
\begin{align*}
\left(\mathbb E(W_{t}^{(j)})\right) A_n = & A_n\\
\mathbb E \left[ W_t^{(j)} \left( W_{t}^{(j)}\right)^{\top} \right] = &\alpha_1 I_n + (1-\alpha_1)A_n \numberthis \label{def_alpha1}\\
\mathbb E \left[ W_t^{(j)} A_n\left( W_{t}^{(j)}\right)^{\top} \right] = &\alpha_2 I_n + (1-\alpha_2)A_n  \numberthis \label{def_alpha2}
\end{align*}
for some constants $\alpha_1$ and $\alpha_2$ satisfying $0<\alpha_2<\alpha_1<1$ (see Lemmas~\ref{lem: EW}, \ref{lem: EWW}, and \ref{lem: EWAW} in Supplementary Material). Since the exact expression is too complex, we plot the $\alpha_1$ and $\alpha_2$ related to different $n$ in Figure~\ref{fig_alpha2} and  Figure~\ref{fig_beta} (detailed discussion is included in \textbf{Section ~\ref{secD}} in Supplementary Material.). Here, we do not plot $\alpha_2$, but plot $\alpha_1-\alpha_2$ instead. This is because $\alpha_1-\alpha_2$ is an important factor in our Theorem (See Section \ref{sec:theorem} where we define $\alpha_1-\alpha_2$ as $\beta$).

\section{Theoretical Guarantees and Discussion} \label{sec:theorem}


Below we show that, for certain parameter values, RPS with unreliable communication rates admits the same convergence rate as the standard algorithms. In other words, the impact of network unreliablity may be seen as negligible.

First let us make some necessary assumptions, that are commonly used in analyzing stochastic optimization algorithms.

\begin{assumption}
\label{ass:global}
We make the following commonly used assumptions:
\begin{enumerate}
  \item \textbf{Lipschitzian gradient:} All function $f_i(\cdot)$'s are with $L$-Lipschitzian gradients, which means
  \begin{align*}
  \|\nabla f_i(\bm{x}) - \nabla f_i(\bm{y})\| \leq L\|\bm{x} - \bm{y}\|
  \end{align*}
 \item \textbf{Bounded variance:} Assume the variance of stochastic gradient
\begin{align*}
    \mathbb{E}_{\xi\sim \mathcal{D}_i} \left\| \nabla F_i (\bm{x}; \xi) - \nabla f_i (\bm{x})\right\|^2 \leqslant & \sigma^2, \quad \forall i, \forall \bm{x},\\
     {1\over n}\sum_{i=1}^n\left\| \nabla f_i (\bm{x})-\nabla f (\bm{x})\right\|^2 \leqslant & \zeta^2, \quad \forall i, \forall \bm{x},
\end{align*}
  is bounded for any $x$ in each worker $i$.
  \item \textbf{Start from 0:} We assume $X_1 = 0$ for simplicity w.l.o.g.
  \end{enumerate}
\end{assumption}


Next we are ready to show our main result.
\begin{theorem}[Convergence of Algorithm~\ref{alg1}] \label{theo:1}
Under Assumption~\ref{ass:global}, choosing $\gamma$ in Algorithm~\ref{alg1} to be small enough that satisfies $1- \frac{6L^2\gamma^2}{(1-\sqrt{\beta})^2}>0$, we have the following convergence rate for Algorithm~\ref{alg1}
\begin{align*}
\frac{1}{T}\sum_{t=1}^T\left( \mathbb{E}\left\|\nabla f(\overline{\bm{x}}_{t})\right\|^2 + (1-L\gamma)\mathbb{E}\left\|\overline{\nabla} f(X_t)\right\|^2 \right)
\leq &\frac{2f(\bm{0}) -2f(\bm{x}^*)}{\gamma T}  + \frac{\gamma L\sigma^2 }{n} + 4\alpha_2 L\gamma(\sigma^2 + 3\zeta^2)\\
 &+ \frac{\left(2\alpha_2 L\gamma + L^2\gamma^2 + 12\alpha_2 L^3 \gamma^3\right)\sigma^2 C_1}{(1-\sqrt{\beta})^2}\\
 &+ \frac{3\left(2\alpha_2 L\gamma + L^2\gamma^2 + 12\alpha_2 L^3 \gamma^3\right)\zeta^2 C_1}{(1-\sqrt{\beta})^2}
,\numberthis\label{theo1eq}
\end{align*}
where
\begin{alignat*}{2}
\nabla f(\overline{\bm{x}}_t) = & f\left(\frac{1}{n}\sum_{i=1}^n\bm{x}_t^{(i)}\right) ,\\
  \overline{\nabla} f(X_t) = &\sum_{i=1}^n\nabla f_i\left(\bm{x}_t^{(i)}\right),\\
 \beta =& \alpha_1 - \alpha_2, \\
  C_1 =& \left(1- \frac{6L^2\gamma^2}{(1-\sqrt{\beta})^2} \right)^{-1},
\end{alignat*}
and $\alpha_1$, $\alpha_2$ follows the definition in \eqref{def_alpha1} and \eqref{def_alpha2}.
\end{theorem}

To make the result more clear, we appropriately choose the learning rate as follows:
\begin{corollary}\label{coro1}
Choose $\gamma = \frac{1-\sqrt{\beta}}{6L + 3(\sigma+\zeta)\sqrt{\alpha_2 T} + \frac{\sigma\sqrt{T}}{\sqrt{n}}}$ in Algorithm~\ref{alg1}, under Assumption~\ref{ass:global}, we have the follow convergence rate for Algorithm~\ref{alg1}
\begin{align*}
\frac{1}{T}\sum_{t=1}^T\mathbb{E}\left\|\nabla f(\overline{\bm{x}}_{t})\right\|^2 
\lesssim &\frac{(\sigma + \zeta)\left(1 + \sqrt{n\alpha_2}\right)}{(1-\sqrt{\beta})\sqrt{nT}} 
 + \frac{1}{T}
 + \frac{n(\sigma^2 + \zeta^2)}{(1 + n\alpha_2  )\sigma^2 T + n\alpha_2 T \zeta^2},
\end{align*}
where $\beta$, $\alpha_1$, $\alpha_2$, $\nabla f(\overline{\bm{x}})$ follow to the definitions in Theorem~\ref{theo:1}, and we treat $f(0)$,$f^*$, and $L$ to be constants.
\end{corollary}

We discuss our theoretical results below
\begin{itemize}
\item ({\bf Comparison with centralized SGD and decentralized SGD}) The dominant term in the convergence rate is $O(1/\sqrt{nT})$ (here we use $\alpha_2 = \mathcal{O}(p(1-p)/n)$ and $\beta= \mathcal{O}(p)$ which is shown by Lemma~\ref{lem: EWAW} in Supplement), which is consistent with the rate for centralized SGD and decentralized SGD \cite{lian2017can}.
\item ({\bf Linear Speedup}) Since the the leading term of convergence rate for $\frac{1}{T}\sum_{t=1}^T\mathbb{E}\left\|\nabla f(\overline{\bm{x}}_{t})\right\|^2$ is $\mathcal{O}(1/\sqrt{nT})$. It suggests that our algorithm admits the linear speedup property with respect to the number of workers $n$.

\item ({\bf Better performance for larger networks}) Fixing the package drop rate $p$ (implicitly included in \textbf{Section ~\ref{secD}}), the convergence rate under a larger network (increasing $n$) would be superior, because the leading terms' dependence of the $\alpha_2 = \mathcal{O}(p(1-p)/n)$. This indicates that the affection of the package drop ratio diminishes, as we increase the number of workers and parameter servers.

\item ({\bf Why only converges to a ball of a critical point}) This is because we use a constant learning rate, the algorithm could only converges to a ball centered at a critical point. This is a common choice to make the statement simpler, just like many other analysis for SGD. Our proved convergence rate is totally consistent with SGD’s rate, and could converge (in the same rate) to a critical point by choosing a decayed learning rate such as $O(1/\sqrt{T})$ like SGD. 

\end{itemize}


\section{Experiments: Convergence of RPS} \label{sec:experiments}

We now validate empirically the scalability and accuracy 
of the RPS algorithm, given reasonable message arrival rates. 

\subsection{Experimental Setup}
\paragraph{Datasets and models} We evaluate our algorithm on two state of the art machine learning tasks: (1) image classification and (2) natural language understanding (NLU). We train ResNet~\cite{he2016deep} with different number of layers on CIFAR-10~\cite{krizhevsky2009learning} for classifying images. We perform the NLU task on the Air travel information system (ATIS) corpus on a one layer LSTM network. 

\paragraph{Implementation} We simulate packet losses by adapting the latest version 2.5 of the Microsoft Cognitive Toolkit~\cite{seide2016cntk}. We implement the RPS algorithm using MPI. During training, we use a local batch size of 32 samples per worker for image classification. We adjust the learning rate by applying a linear scaling rule~\cite{goyal2017accurate} and decay of 10 percent after 80 and 120 epochs, respectively. 
To achieve the best possible convergence, we apply a gradual warmup strategy~\cite{goyal2017accurate} during the first 5 epochs. We deliberately do not use any regularization or momentum during the experiments in order to be consistent with the described algorithm and proof. 
The NLU experiments are conducted with the default parameters given by the CNTK examples, with scaling the learning rate accordingly, and omit momentum and regularization terms on purpose.
The training of the models is executed on 16 NVIDIA TITAN Xp GPUs. The workers are connected by Gigabit Ethernet. We use each GPU as a worker.
We describe the results in terms of training loss convergence, although the validation trends are similar. 

%

\begin{figure*}[tbp]
	\centering
	\subfloat[ResNet20 - CIFAR10]{
\includegraphics[width=0.33\textwidth,height=\textheight,keepaspectratio]{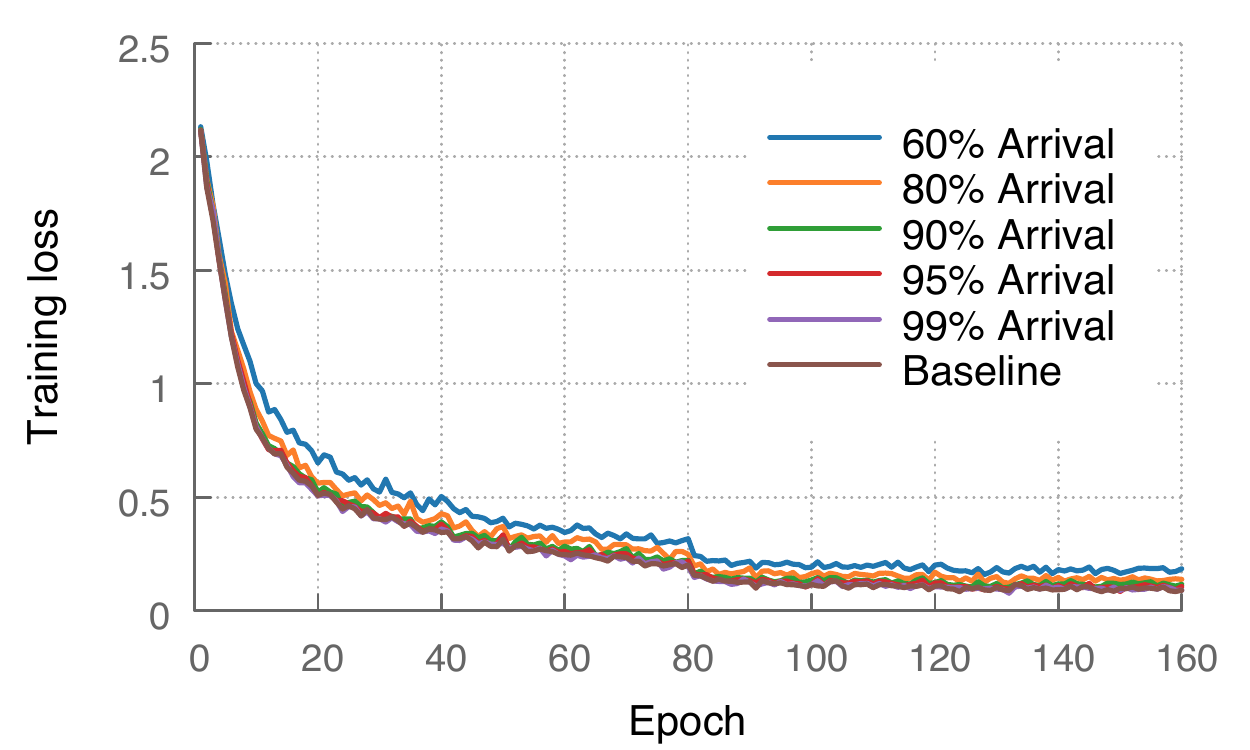}
	}
	\subfloat[ResNet110 - CIFAR10]{
\includegraphics[width=0.33\textwidth,height=\textheight,keepaspectratio]{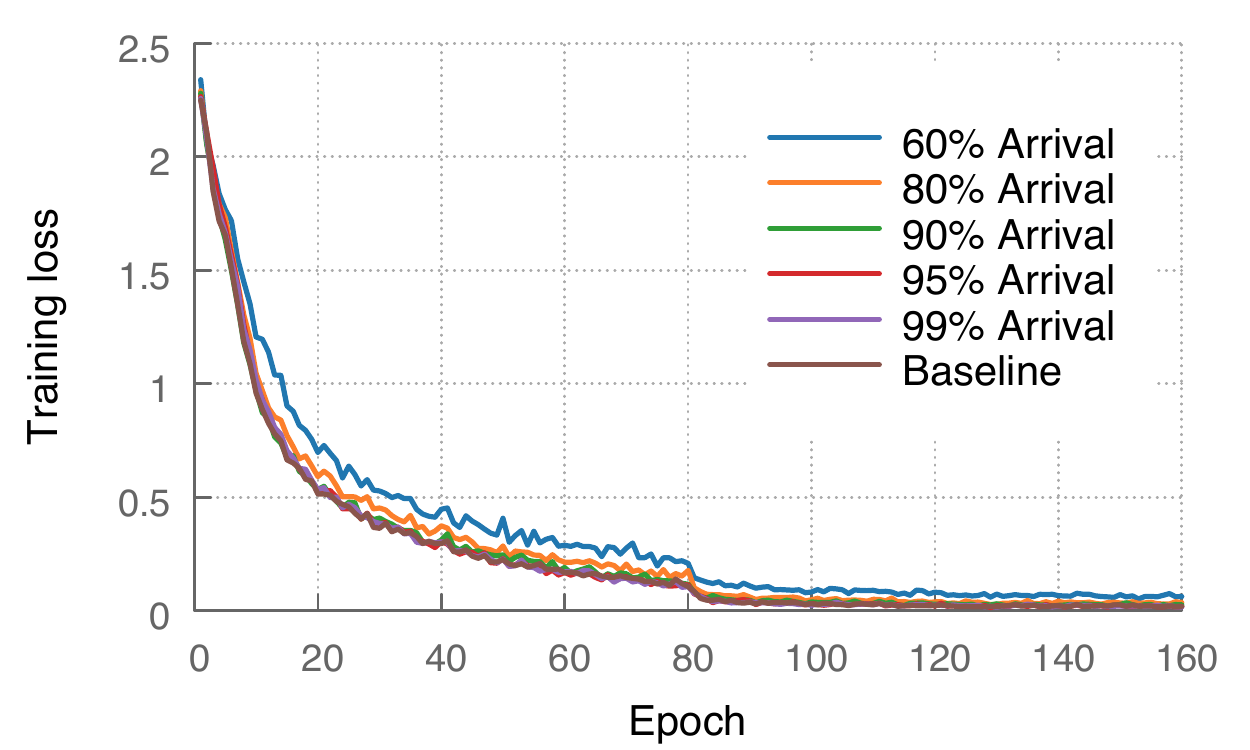}
	}
	\subfloat[LSTM - ATIS]{
\includegraphics[width=0.33\textwidth,height=\textheight,keepaspectratio]{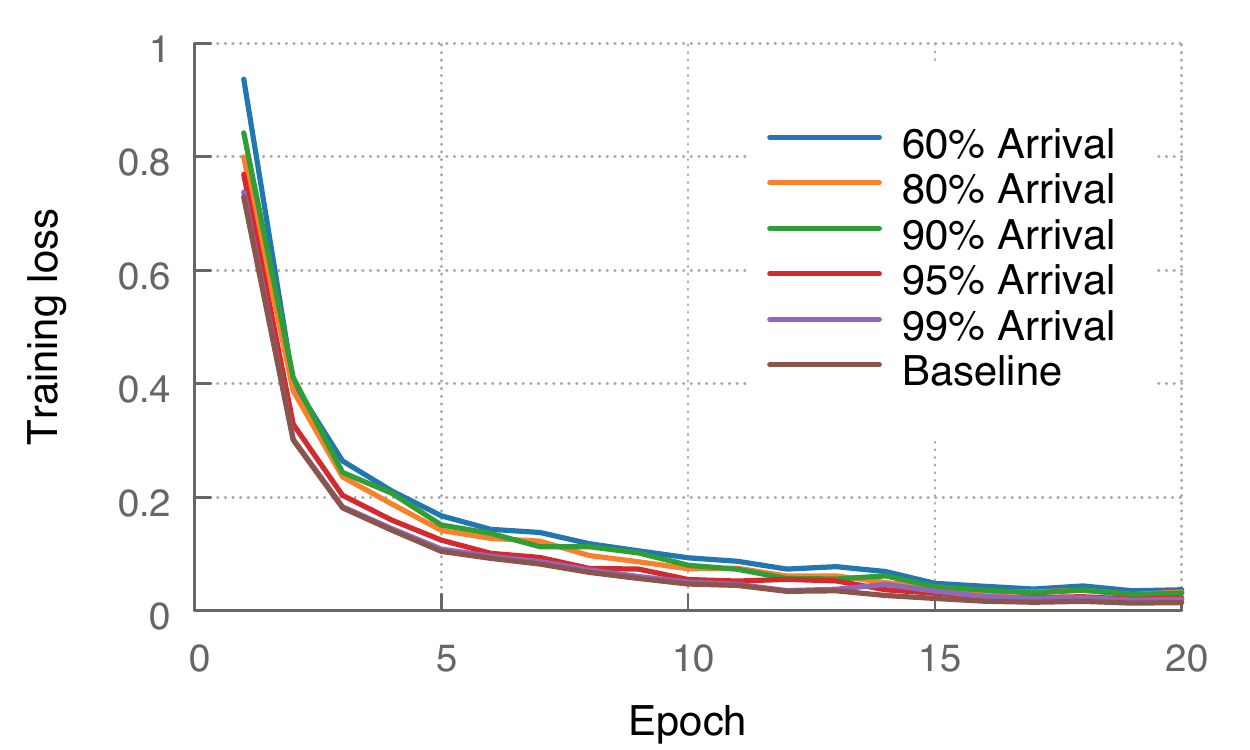}
	}
    \caption{Convergence of RPS on different datasets.}
    \label{fig:model}
\end{figure*}

\paragraph{Convergence of Image Classification} We perform convergence tests using the analyzed algorithm, model averaging SGD, on both ResNet110 and ResNet20 with CIFAR-10. Figure~\ref{fig:model}(a,b) shows
the result. We vary probabilities for each packet being correctly delivered at each worker between 80\%, 90\%, 95\% and 99\%. The baseline is 100\% message delivery rate. The baseline achieves a training loss of 0.02 using ResNet110 and 0.09 for ResNet20. Dropping 1\% doesn't increase the training loss achieved after 160 epochs. For 5\% the training loss is identical on ResNet110 and increased by 0.02 on ResNet20. Having a probability of 90\% of arrival leads to a training loss of 0.03 for ResNet110 and 0.11 for ResNet20 respectively.

\paragraph{Convergence of NLU} We perform full convergence tests for the NLU task on the ATIS corpus and a single layer LSTM.
Figure~\ref{fig:model}(c) shows the result. The baseline achieves a training loss of 0.01. Dropping 1, 5 or 10 percent of the communicated partial vectors result in an increase of 0.01 in training loss.


\begin{figure*}[tbp]
	\centering
	\subfloat[ResNet20 - CIFAR10]{
\includegraphics[width=0.33\textwidth,height=\textheight,keepaspectratio]{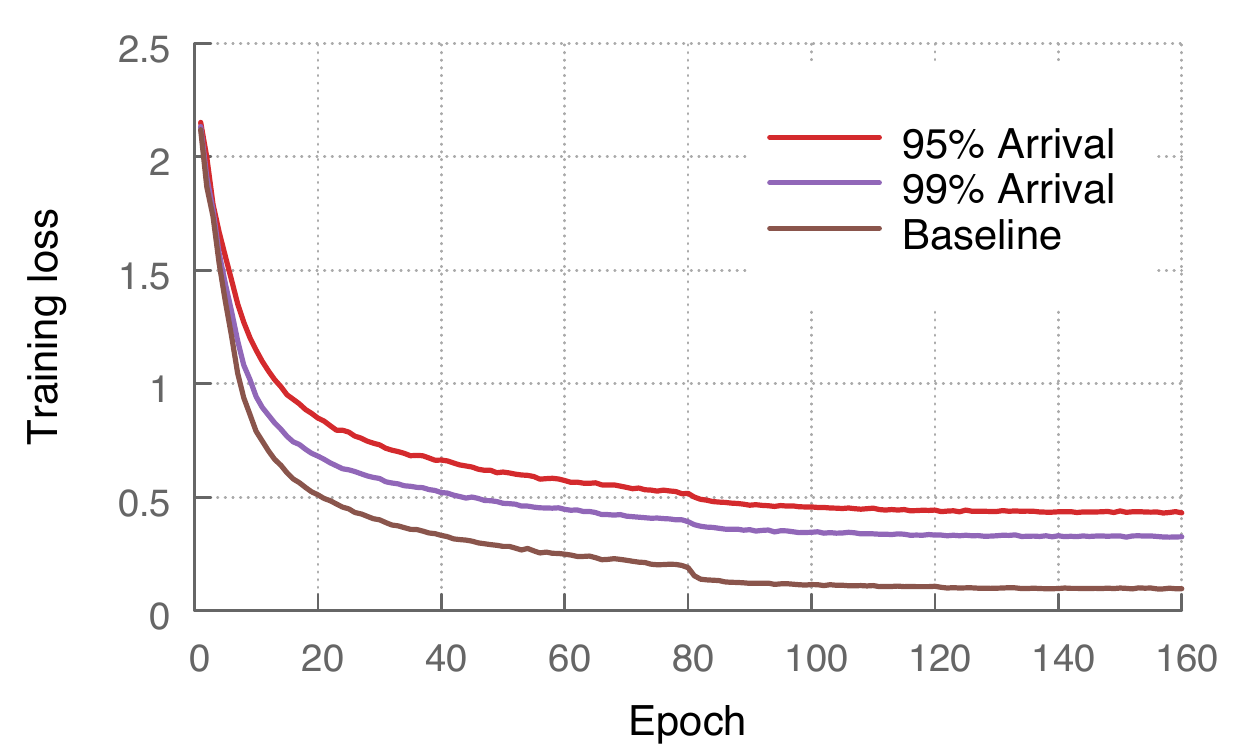}
	}
	\subfloat[ResNet110 - CIFAR10]{
\includegraphics[width=0.33\textwidth,height=\textheight,keepaspectratio]{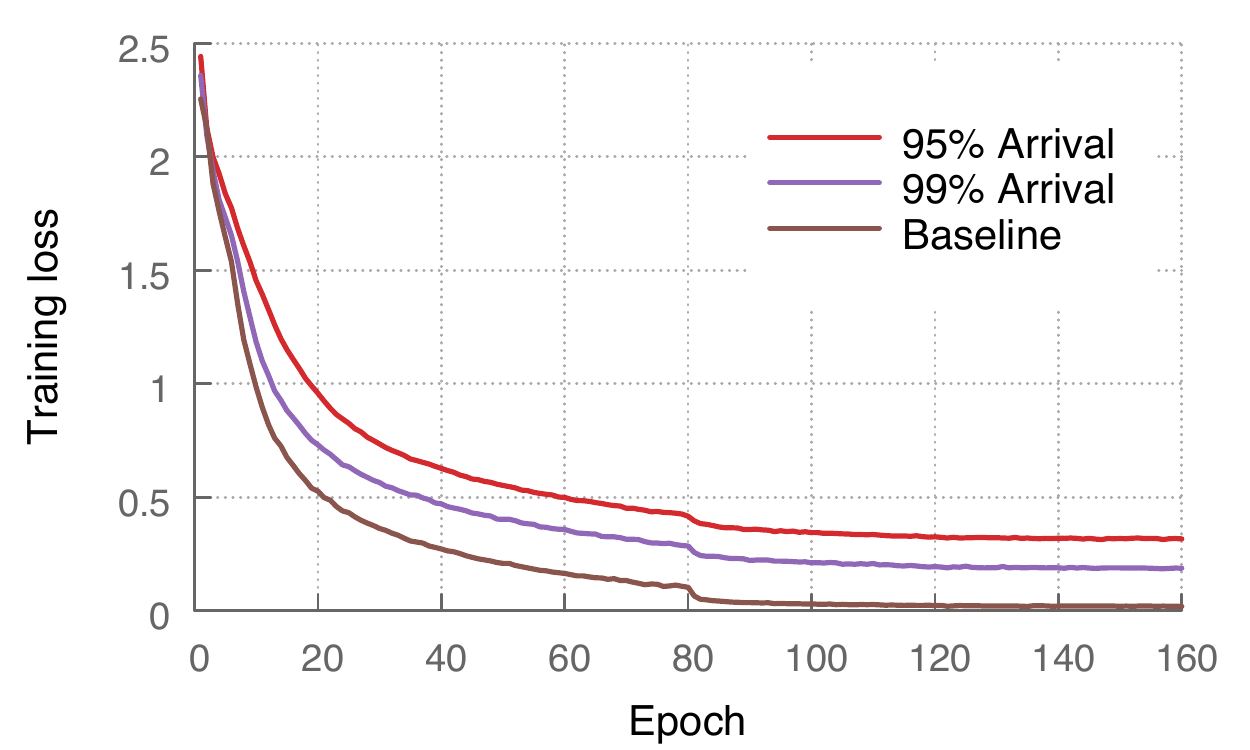}
	}
	\subfloat[LSTM - ATIS]{
\includegraphics[width=0.33\textwidth,height=\textheight,keepaspectratio]{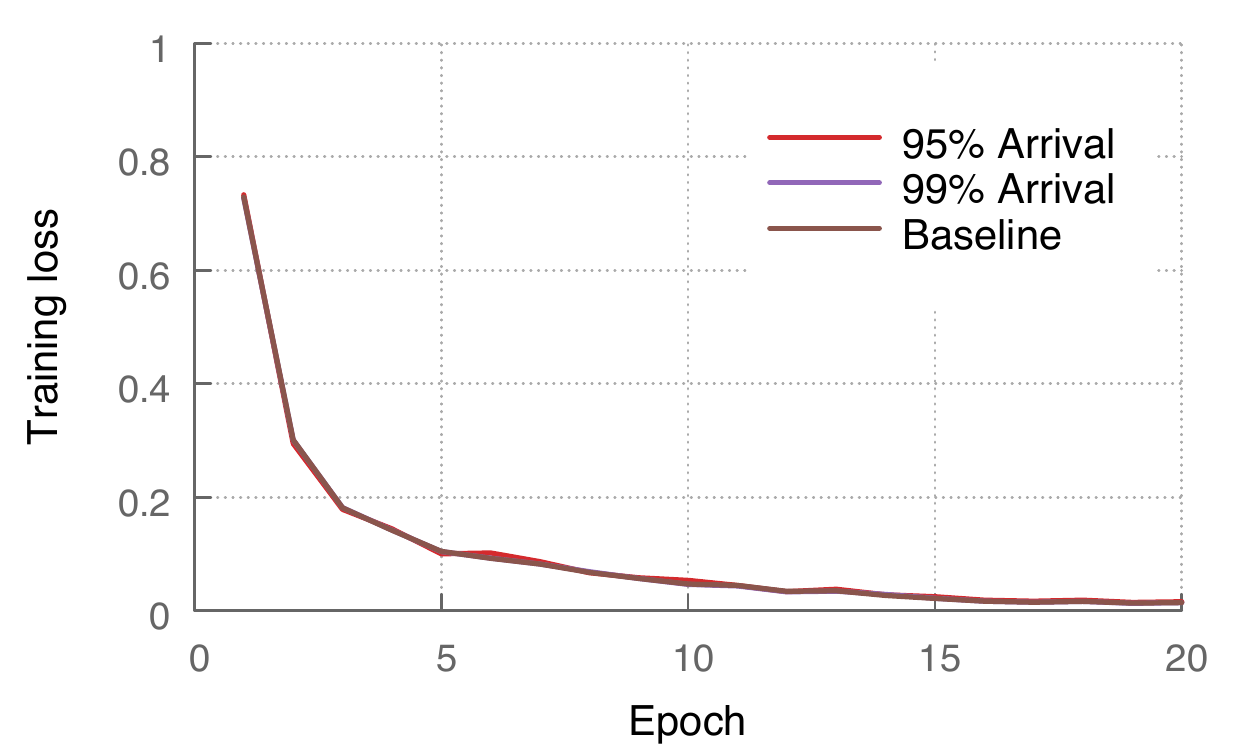}
	}
    \caption{Why RPS? The Behavior of Standard SGD in the Presence of Message Drop.}
    \label{fig:grad}
\end{figure*}

\paragraph{Comparison with Gradient Averaging} We conduct experiments with identical setup and a probability of 99 percent of arrival using a \emph{gradient} averaging methods, instead of model averaging. When running data distributed SGD, gradient averaging is the most widely used technique in practice, also implemented by default in most deep learning frameworks\cite{abadi2016tensorflow, seide2016cntk}. As expected, the baseline (all the transmissions are successful) convergences to the same training loss as its model averaging counterpart, when omitting momentum and regularization terms. As seen in figures~\ref{fig:grad}(a,b), having a loss in communication of even 1 percentage results in worse convergence in terms of accuracy for both ResNet architectures on CIFAR-10. The reason is that the error of package drop will accumulate over iterations but never decay, because the model is the sum of all early gradients, so the model never converges to the optimal one. Nevertheless, this insight suggests that one should favor a model averaging algorithm over gradient averaging, if the underlying network connection is unreliable.

\section{Case study: Speeding up Colocated Applications} \label{sec:case}

\begin{figure}[h]
\centering
\begin{minipage}{.4\textwidth}
  \centering
  \includegraphics[width=1.0\linewidth]{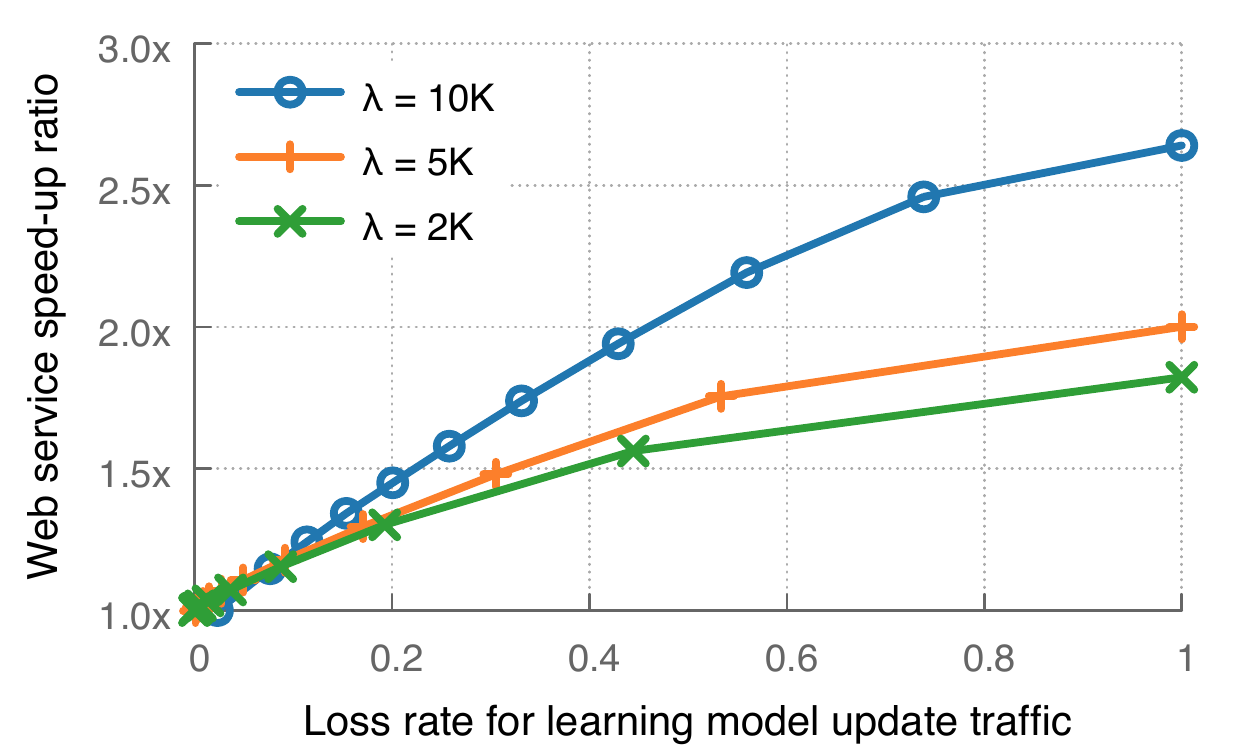}
  \captionof{figure}{Allowing an increasing rate of losses for model updates speeds up the Web service.}
  \label{fig:drop-rate-to-speed-up}
\end{minipage}%
\hspace{0.05\textwidth}%
\begin{minipage}{.4\textwidth}
  \centering
  \includegraphics[width=1.0\linewidth]{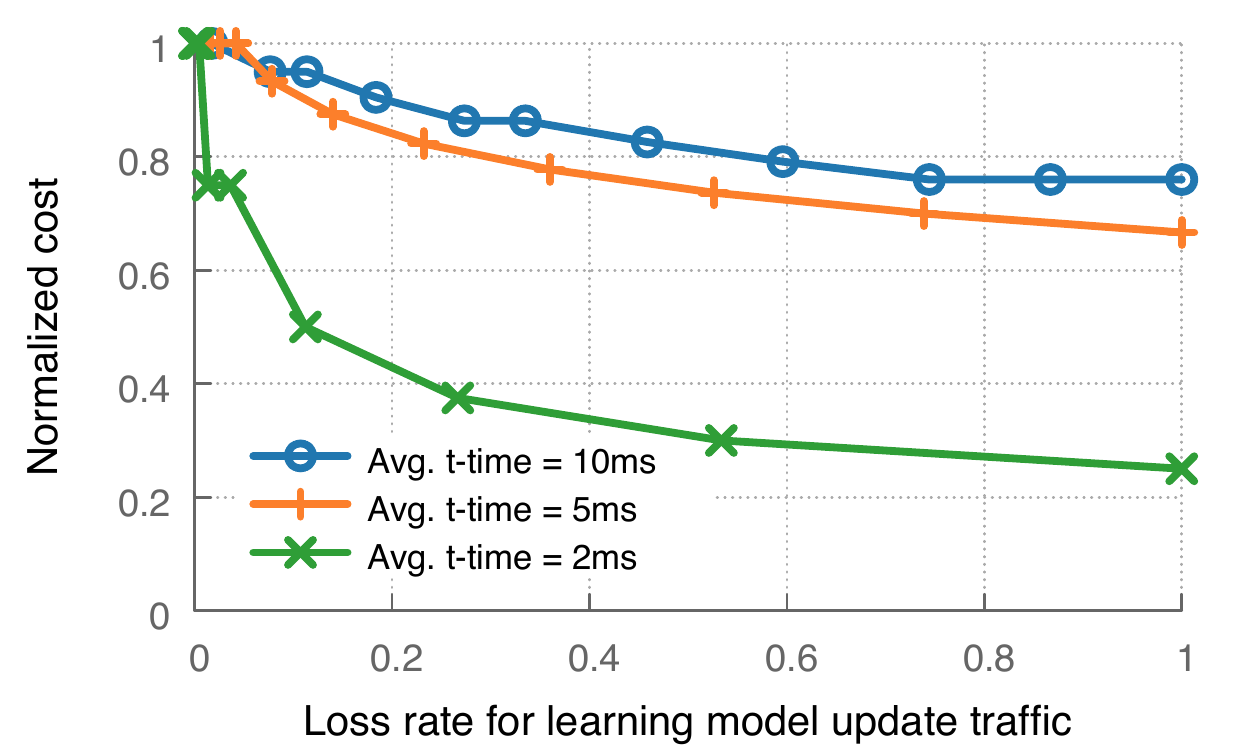}
  \captionof{figure}{Allowing more losses for model updates reduces the cost for the Web service.}
  \label{fig:drop-rate-to-normalized-cost}
\end{minipage}%
\end{figure}

Our results on the resilience of distributed learning to losses of model updates open up an interesting use case. That model updates can be lost (within some tolerance) without the deterioration of model convergence implies that model updates transmitted over the physical network can be de-prioritized compared to other more ``inflexible,'' delay-sensitive traffic, such as for Web services. 
Thus, we can colocate other applications with the training workloads, and reduce infrastructure costs for running them. Equivalently, workloads that are colocated with learning workers can benefit from prioritized network traffic (at the expense of some model update losses), and thus achieve lower latency.


To demonstrate this in practice, we perform a packet simulation over 16 servers, each connected with a $1$~Gbps link to a network switch. 
Over this network of $16$ servers, we run two workloads: (a) replaying traces from the machine learning process of ResNet110 on CIFAR-10 (which translates to a load of 2.4 Gbps) which is sent \emph{unreliably}, and (b) a simple emulated Web service running on all $16$ servers. Web services often produce significant background traffic between servers within the data center, consisting typically of small messages fetching distributed pieces of content to compose a response (e.g., a Google query response potentially consists of advertisements, search results, and images). We emulate this intra data center traffic for the Web service as all-to-all traffic between these servers, with small messages of $100$~KB (a reasonable size for such services) sent reliably between these servers. The inter-arrival time for these messages follows a Poisson process, parametrized by the expected message rate, $\lambda$ (aggregated across the $16$ servers).  

Different degrees of prioritization of the Web service traffic over learning traffic result in different degrees of loss in learning updates transmitted over the network. As the Web service is prioritized to a greater extent, its performance improves -- its message exchanges take less time; we refer to this reduction in (average) completion time for these messages as a speed-up. Note that even small speedups of $10\%$ are significant for such services; for example, Google actively pursues minimizing its Web services' response latency. An alternative method of quantifying the benefit for the colocated Web service is to measure how many additional messages the Web service can send, while maintaining a fixed average completion time. This translates to running more Web service queries and achieving more throughput over the same infrastructure, thus reducing cost per request / message.



Fig.~\ref{fig:drop-rate-to-speed-up} and Fig. \ref{fig:drop-rate-to-normalized-cost} show results for the above described Web service speedup and cost reduction respectively. In Fig.~\ref{fig:drop-rate-to-speed-up}, the arrival rate of Web service messages is fixed ($\lambda = \{2000, 5000, 10000\}$ per second). As the network prioritizes the Web service more and more over learning update traffic, more learning traffic suffers losses (on the $x$-axis), but performance for the Web service improves. With just $10\%$ losses for learning updates, the Web service can be sped up by more than $20\%$ (\emph{i.e.,} $1.2\times$). 

In Fig.~\ref{fig:drop-rate-to-normalized-cost}, we set a target average transmission time ($2$, $5$, or $10$~ms) for the Web service's messages, and increase the message arrival rate, $\lambda$, thus causing more and more losses for learning updates on the $x$-axis. But accommodating higher $\lambda$ over the same infrastructure translates to a lower cost of running the Web service (with this reduction shown on the $y$-axis).

Thus, tolerating small amounts of loss in model update traffic can result in significant benefits for colocated services, while not deteriorating convergence.

\section{Conclusion} \label{sec:conclusion}

In this paper, we present a novel analysis for a general model of distributed machine learning, 
under a realistic unreliable communication model.
We present a novel theoretical analysis for such a scenario, and
evaluated it while training neural networks on both image and natural language datasets. 
We also provided a case study of application collocation, to illustrate the potential benefit that can be provided
by allowing learning algorithms to take advantage of unreliable communication channels.

\bibliography{references}
\bibliographystyle{abbrvnat}

\newpage
\input{supp}

\end{document}

%% file: supp.tex
\appendix
\newpage
\onecolumn
\begin{center}
{\Huge \bf
Supplemental Materials
}
\end{center}

\section {Notations}
In order to unify notations, we define the following notations about gradient:
\begin{align*}
\bm{g}^{(i)}(\bm{x}^{(i)}_t;\bm{\xi}^{(i)}_t) = \nabla F_i(\bm{x}^{(i)}_t;\bm{\xi}^{(i)}_t)
\end{align*}
We define $I_n$ as the $n\times n$ identity matrix, $\mathbf{1}_n$ as $(1,1,\cdots, 1)^\top$ and $A_n$ as $\frac{1}{n}\mathbf{1}\mathbf{1}^{\top}$. Also, we suppose the packet drop rate is $p$.

The following equation is used frequently:
\begin{equation} \label{eq: T}
	\Tr (XA_nX^{\top}) = \Tr\Big(X\frac{\mathbf{1}\mathbf{1}^{\top}}{n}X^{\top}\Big) = n\Tr \bigg(\Big( X\frac{\mathbf{1}}{n}\Big)^{\top}X\frac{\mathbf{1}}{n}\bigg) = n\Big( X\frac{\mathbf{1}}{n}\Big)^{\top}X\frac{\mathbf{1}}{n} = n\left\|X\frac{\mathbf{1}}{n}\right\|^2
\end{equation}

\subsection{Matrix Notations}
We aggregate vectors into matrix, and using matrix to simplify the proof.
\begin{align*}
	X_t &:= \left(\bm{x}_t^{(1)}, \bm{x}_t^{(2)}, \cdots, \bm{x}_t^{(n)}\right)\\
	V_t &:= \left(\bm{v}_t^{(1)}, \bm{v}_t^{(2)}, \cdots, \bm{v}_t^{(n)}\right)\\
	\Xi_t &:= \big(\bm{\xi}_t^{(1)}, \bm{\xi}_t^{(2)}, \cdots, \bm{\xi}_t^{(n)}\big)\\
	G(X_t;\Xi_t)&:= \left(\bm{g}^{(1)}(\bm{x}_t^{(1)};\bm{\xi}_t^{(1)}),\cdots,\bm{g}^{(n)}(\bm{x}_t^{(n)};\bm{\xi}_t^{(n)}) \right)
\end{align*}

\subsection{Averaged Notations}
We define averaged vectors as follows:
\begin{align*}
	\overline{\bm{x}}_t &:= \frac{1}{n}\sum_{i=1}^n \bm{x}_t^{(i)} \numberthis\label{re: 1}\\
	\overline{\bm{v}}_t &:= \frac{1}{n}\sum_{i=1}^{n} \bm{v}_t \numberthis\label{re: 7}\\
	\overline{\bm{g}}(X_t;\Xi_t) &:= \sum_{i=1}^n \bm{g}^{(i)}(\bm{x}^{(i)}_t;\bm{\xi}^{(i)}_t) \numberthis\label{re: 8}\\
	\overline{\nabla}f(X_t)&:= \frac{1}{n}\sum\limits_{i=1}^n f_i(\bm{x}_t^{(i)})\\
	\Delta \overline{\bm{x}}_t &:= \overline{\bm{x}}_{t+1} - \overline{\bm{x}}_t
\end{align*}

\subsection{Block Notations}
Remember in (\ref{eq: dividev}) and (\ref{eq: dividex}), we have divided models in blocks:
\begin{align*}
\bm{v}_t^{(i)} = &\left(\left(\bm{v}_t^{(i,1)}\right)^{\top}, \left(\bm{v}_t^{(i,2)}\right)^{\top}, \cdots, \left(\bm{v}_t^{(i,n)}\right)^{\top}\right)^{\top}\\
\bm{x}_{t}^{(i)} = &\left(\left(\bm{x}_{t}^{(i,1)}\right)^{\top}, \left(\bm{x}_{t}^{(i,2)}\right)^{\top},\cdots, \left(\bm{x}_{t}^{(i,n)}\right)^{\top}\right)^{\top}, \forall i\in[n].
\end{align*}
 We do the some division on some other quantities, see following (the dimension of each block is the same as the corresponding block in $\bm{v}_t^{(i)}$) : 
\begin{align*}
	\overline{\bm{x}}_t &= \left(\left(\overline{\bm{x}}_t^{(1)}\right)^{\top},\left(\overline{\bm{x}}_t^{(2)}\right)^{\top},\cdots, \left(\overline{\bm{x}}_t^{(n)}\right)^{\top}\right)^{\top}\\
	\overline{\bm{v}}_t &= \left(\left(\overline{\bm{v}}_t^{(1)}\right)^{\top},\left(\overline{\bm{v}}_t^{(2)}\right)^{\top},\cdots, \left(\overline{\bm{v}}_t^{(n)}\right)^{\top}\right)^{\top}\\
	\Delta \overline{\bm{x}}_t &= \left(\left(\Delta^{(1)} \overline{\bm{x}}_t\right)^{\top}, \left(\Delta_t^{(2)} \overline{\bm{x}}\right)^{\top}, \cdots, \left(\Delta_t^{(n)} \overline{\bm{x}}\right)^{\top}\right)^{\top}\\
	\bm{g}^{(i)}(\cdot; \cdot) &= \left(\left(\bm{g}^{(i,1)}(\cdot; \cdot)\right)^{\top}, \left(\bm{g}^{(i,2)}(\cdot; \cdot)\right)^{\top}, \cdots, \left(\bm{g}^{(i,n)}(\cdot; \cdot)\right)^{\top}\right)^{\top}\\
	\overline{\bm{g}}(X_t;\Xi_t) &= \left(\left(\overline{\bm{g}}^{(1)}(X_t;\Xi_t)\right)^{\top},\cdots, \left(\overline{\bm{g}}^{(n)}(X_t;\Xi_t)\right)^{\top}\right)^{\top}\\
	\nabla f_i(\bm{x}_t^{(i)}) &= \left(\left(\nabla^{(1)} f_i(\bm{x}_t^{(i)})\right)^{\top}, \left(\nabla^{(2)} f_i(\bm{x}_t^{(i)})\right)^{\top}, \cdots, \left(\nabla^{(n)} f_i(\bm{x}_t^{(i)})\right)^{\top}\right)^{\top}\\
	\overline{\nabla} f(X_t) &= \left(\left(\overline{\nabla}^{(1)} f(X_t)\right)^{\top}, \left(\overline{\nabla}^{(2)} f(X_t)\right)^{\top}, \cdots, \left(\overline{\nabla}^{(n)} f(X_t)\right)^{\top}\right)^{\top}\\
	\nabla f(\overline{\bm{x}}_t)&= \left(\left(\nabla^{(1)} f(\overline{\bm{x}}_t)\right)^{\top}, \left(\nabla^{(2)} f(\overline{\bm{x}}_t)\right)^{\top}, \cdots, \left(\nabla^{(n)} f(\overline{\bm{x}}_t)\right)^{\top}\right)^{\top}.
\end{align*}

\subsection{Aggregated Block Notations}
Now, we can define some additional notations throughout the following proof
\begin{align*}
X_t^{(j)}:=&(\bm{x}_t^{(1,j)},\bm{x}_t^{(2,j)},\cdots,\bm{x}_t^{(n,j)})\\
V_t^{(j)}:=&(\bm{v}_t^{(1,j)},\bm{v}_t^{(2,j)},\cdots,\bm{v}_t^{(n,j)})\\
G^{(j)}(X_t;\Xi_t):= &\left(\bm{g}^{(1,j)}(\bm{x}_t^{(1)};\bm{\xi}_t^{(1)}),\cdots,\bm{g}^{(n,j)}(\bm{x}_t^{(n)};\bm{\xi}_t^{(n)}) \right)
\end{align*}

\subsection{Relations between Notations}
We have the following relations between these notations:
\begin{align*}
	\overline{\bm{x}}_t^{(j)} &= X_t^{(j)}\frac{\mathbf{1}}{n} \numberthis\label{re: 2}\\
	\overline{\bm{v}}_t^{(j)} &= V_t^{(j)}\frac{\mathbf{1}}{n} \numberthis\label{re: 3}\\
	\overline{\bm{g}}^{(j)}(X_t;\Xi_t) &= G^{(j)}(X_t;\Xi_t)\frac{\mathbf{1}}{n} \numberthis\label{re: 9}\\
	A_nA_n &= A_n \numberthis\label{re: 4}\\ 
	V_t &= X_t - \gamma G(X_t; \Xi_t) \numberthis\label{re: 5}\\
	V_t^{(j)} &= X_t^{(j)} - \gamma G^{(j)}(X_t; \Xi_t) \numberthis\label{re: 6}
\end{align*}

\subsection{Expectation Notations}
There are different conditions when taking expectations in the proof, so we list these conditions below:
\paragraph{$\mathbb{E}_{t,G}[\cdot]$}
Denote taking the expectation over the \textbf{computing stochastic Gradient} procedure at $t$th iteration on condition of the history information before the $t$th iteration.
\paragraph{$\mathbb{E}_{t,P}[\cdot]$}
Denote taking the expectation over the \textbf{Package drop in sending and receiving blocks} procedure at $t$th iteration on condition of the history information before the $t$th iteration and the SGD procedure at the $t$th iteration.
\paragraph{$\mathbb{E}_t[\cdot]$}
Denote taking the expectation over all procedure during the $t$th iteration on condition of the history information before the $t$th iteration.
\paragraph{$\mathbb{E}[\cdot]$}
Denote taking the expectation over all history information.

\section{Proof to Theorem~\ref{theo:1}}
The critical part for a decentralized algorithm to be successful, is that local model on each node will converge to their average model. We summarize this critical property by the next lemma.

\begin{lemma}\label{L:xavekey}
From the updating rule (\ref{eq: updatingrule}) and Assumption \ref{ass:global}, we have
\begin{align*}
\sum_{s=1}^T\sum_{i=1}^n\mathbb E\left\|\bm{x}_{s+1}^{(i)} - \overline{\bm{x}}_{s+1}\right\|^2 \leq & \frac{2\gamma^2n\sigma^2TC_1}{(1-\sqrt{\beta})^2} + \frac{6n\zeta^2TC_1}{(1-\sqrt{\beta})^2},\numberthis\label{lemma:xavekey}
\end{align*}
where $C_1:=\left(1- \frac{6L^2\gamma^2}{(1-\sqrt{\beta})^2} \right)^{-1}$ and $\beta = \alpha_1 - \alpha_2 $.
\end{lemma}

We will prove this critical property first. Then, after proving some lemmas, we will prove the final theorem. During the proof, we will use properties of weighted matrix $W_t^{(j)}$ which is showed in \textbf{Section ~\ref{secD}}.

\subsection{Proof of Lemma~\ref{L:xavekey}} 
\begin{proof} [\textbf{Proof to Lemma~\ref{L:xavekey}}]
According to updating rule (\ref{eq: updatingrule}) and Assumption \ref{ass:global}, we have
\begin{align*}
X_{t+1}^{(j)} = & V_t^{(j)}W_{t}^{(j)}\\
= & \left(X_t^{(j)} - \gamma G^{(j)}(X_t;\Xi_t)\right)W_{t}^{(j)}\\
= & X_1^{(j)}\prod_{r=1}^t W_{r}^{(j)} - \gamma\sum_{s=1}^tG^{(j)}(X_s;\Xi_s)\prod_{r=s}^tW_{r}^{(j)}\\
= & - \gamma\sum_{s=1}^tG^{(j)}(X_s;\Xi_s)\prod_{r=s}^tW_{r}^{(j)}. \text{   } \left(\text{due to } X_1 = 0\right)\numberthis\label{lemma:xavekey_1} 
\end{align*}

We also have 
\begin{align*}
\sum_{i=1}^n\left\|\bm{x}_{t+1}^{(i,j)} - \overline{\bm{x}}_{t+1}^{(j)}\right\|^2 = &\left\|X_{t+1}^{(j)} - X_{t+1}^{(j)}\frac{\bm{1}}{n}\bm{1}^{\top}_n \right\|^2_F =  \left\|X_{t+1}^{(j)} - X_{t+1}^{(j)}A_n \right\|^2_F \numberthis\label{lemma:xavekey_2}
\end{align*}
Combing \eqref{lemma:xavekey_1} and \eqref{lemma:xavekey_2} together, and define
\begin{align*}
H_{t,s}^{(j)} : = G^{(j)}(X_s;\Xi_s)\prod_{r=s}^tW_{r}^{(j)}
\end{align*}
we get
\begin{align*}
\sum_{i=1}^n\left\|\left(\bm{x}_{t+1}^{(i,j)} - \overline{\bm{x}}_{t+1}^{(j)}\right)\right\|^2 = & \left\|X_{t+1}^{(j)}(I_n - A_n) \right\|^2_F\\
= & \gamma^2\left\|\sum_{s=1}^tH_{t,s}^{(j)}(I_n-A_n) \right\|^2_F\\
= & \gamma^2 \Tr\left((I_n-A_n)\sum_{s=1}^t\left(H_{t,s}^{(j)}\right)^{\top}\sum_{s'=1}^tH_{t,s'}^{(j)}(I-An) \right)\\
= & \gamma^2 \sum_{s,s'=1}^t\Tr\left((I_n-A_n)\left(H_{t,s}^{(j)}\right)^{\top}H_{t,s'}^{(j)}(I-An) \right)\\
\leq & \frac{\gamma^2}{2} \sum_{s,s'=1}^t\left( k_{s,s'}\left\|H_{t,s}^{(j)}(I_n-A_n)\right\|^2_F + \frac{1}{k_{s,s'}}\left\|H_{t,s'}^{(j)}(I_n-A_n)\right\|^2_F \right) ,\numberthis\label{lemma:xavekey_3}
\end{align*}
where $k_{s,s'}$ is a scale factor that is to be computed later. The last inequality is because $2\Tr(A^{\top}B)\leq k\|A\|_F^2 + \frac{1}{k}\|B\|_F^2$ for any matrix $A$ and $B$.

For $\left\|H_{t,s}^{(j)}(I_n-A_n)\right\|^2_F$, we have
\begin{align*}
&\mathbb E \left\|H_{t,s}^{(j)}(I_n-A_n)\right\|^2_F\\
 = & \mathbb E\Tr\left( G^{(j)}(X_s;\Xi_s)W_{s}^{(j)}\cdots W_{t}^{(j)} (I_n-A_n) \left(W_{t}^{(j)}\right)^{\top}\cdots  \left(W_{s}^{(j)}\right)^{\top} \left(G^{(j)}(X_s;\Xi_s)\right)^{\top}  \right). \numberthis\label{lemma: keyyc_1}
\end{align*}

Now we can take expectation from time $t-1$ back to time $s-1$. When taking expectation of time $t$, we only need to compute $\mathbb{E} \left[W_t^{(j)}(I_n-A_n)\left(W_t^{(j)}\right)^{\top}\right]$. From Lemma \ref{lem: EWW} and Lemma \ref{lem: EWAW}, this is just $(\alpha_1 - \alpha_2)(I_n-A_n)$. Applying this to (\ref{lemma: keyyc_1}), we can get the similar form except replacing $t$ by $t-1$ and multiplying by factor $\alpha_1-\alpha_2$. Therefore, we have the following:
\begin{align*}
\mathbb E \left\|H_{t,s}^{(j)}(I_n-A_n)\right\|^2_F = & (\alpha_1- \alpha_2)^{t-s+1}\mathbb{E} \left\|G^{(j)}(X_s;\Xi_s)(I_n-A_n)\right\|^2_F\\
 \leq & (\alpha_1-\alpha_2)^{t-s}\mathbb{E} \left\|G^{(j)}(X_s;\Xi_s)(I_n-A_n)\right\|^2_F.
\end{align*}
The last inequality comes from $\alpha_2\le \alpha_1$c and $\beta = \alpha_1-\alpha_2 $ is defined in Theorem~\ref{theo:1} .

Then \eqref{lemma:xavekey_3} becomes
\begin{align*}
&\sum_{i=1}^n\left\|\left(\bm{x}_{t+1}^{(i,j)} - \overline{\bm{x}}_{t+1}^{(j)}\right)\right\|^2\\
 \leq  & \frac{\gamma^2}{2} \sum_{s,s'=1}^t\left( k_{s,s'}\beta^{t-s}\left\|G^{(j)}(X_s;\Xi_s)(I_n-A_n)\right\|^2_F + \frac{1}{k_{s,s'}}\beta^{t-s'}\left\|G^{(j)}(X_{s'};\xi_{s'})(I_n-A_n)\right\|^2_F \right).
\end{align*}
So if we choose $k_{s,s'} = \beta^{\frac{s-s'}{2}}$, the above inequality becomes
\begin{align*}
&\sum_{i=1}^n\left\|\left(\bm{x}_{t+1}^{(i,j)} - \overline{\bm{x}}_{t+1}^{(j)}\right)\right\|^2\\
 \leq  & \frac{\gamma^2}{2} \sum_{s,s'=1}^t\left( \beta^{\frac{2t-s-s'}{2}}\left\|G^{(j)}(X_s;\Xi_s)(I_n-A_n)\right\|^2_F + \beta^{\frac{2t-s'-s}{2}}\left\|G^{(j)}(X_{s'};\xi_{s'})(I_n-A_n)\right\|^2_F \right)\\
 = & \frac{\gamma^2\beta^{t}}{2} \sum_{s,s'=1}^t\beta^{\frac{-s-s'}{2}}\left( \left\|G^{(j)}(X_s;\Xi_s)(I_n-A_n)\right\|^2_F + \left\|G^{(j)}(X_{s'};\xi_{s'})(I_n-A_n)\right\|^2_F \right)\\
 = & \gamma^2\beta^{t} \sum_{s,s'=1}^t\beta^{\frac{-s-s'}{2}} \left\|G^{(j)}(X_s;\Xi_s)(I_n-A_n)\right\|^2_F \\
 = & \gamma^2 \sum_{s=1}^t\beta^{\frac{t-s}{2}} \left\|G^{(j)}(X_s;\Xi_s)(I_n-A_n)\right\|^2_F \sum_{s'=1}^t \beta^{\frac{t-s'}{2}}\\
 \leq & \frac{\gamma^2}{1-\sqrt{\beta}} \sum_{s=1}^t\beta^{\frac{t-s}{2}} \left\|G^{(j)}(X_s;\Xi_s)(I_n-A_n)\right\|^2_F \numberthis\label{lemma:xavekey_4}
\end{align*}

We also have: 
\begin{align*}
&\sum_{j=1}^n\mathbb E\left\|G^{(j)}(X_s;\Xi_s)(I_n-A_n)\right\|_F^2\\
=& \sum_{j=1}^n\sum_{i=1}^n\mathbb E_{t,G}\left\|\bm{g}^{(i,j)}(\bm{x}_t^{(i)};\bm{\xi}_t^{(i)}) - \overline{\bm{g}}^{(j)}(X_t;\Xi_t)\right\|^2\\
\leq & \sum_{j=1}^n\sum_{i=1}^n\mathbb E_{t,G}\left\|\bm{g}^{(i,j)}(\bm{x}_t^{(i)};\bm{\xi}_t^{(i)}) -\nabla^{(j)} f_i(\bm{x}_t^{(i)})\right\|^2 \\
& + 3\sum_{j=1}^n\sum_{i=1}^n\left\| \nabla^{(j)} f_i(\overline{\bm{x}}_t) - \nabla^{(j)} f(\overline{\bm{x}}_t) \right\|^2 + 6L\sum_{j=1}^n\sum_{i=1}^n\left\|\bm{x}_t^{(i,j)} - \overline{\bm{x}}_t^{(j)}\right\|^2 \text{    } (\text{using } \eqref{lemma:key_4})\\
= & \sum_{i=1}^n\mathbb E_{t,G}\left\|\bm{g}^{(i)}(\bm{x}_t^{(i)};\bm{\xi}_t^{(i)}) -\nabla f^{(i)}(\bm{x}_t^{(i)})\right\|^2 \\
& + 3\sum_{i=1}^n\left\| \nabla f^{(i)}(\overline{\bm{x}}_t) - \nabla f(\overline{\bm{x}}_t) \right\|^2 + 6L\sum_{i=1}^n\left\|\bm{x}_t^{(i)} - \overline{\bm{x}}_t\right\|^2\\
\leq & n\sigma^2 + 6L\sum_{i=1}^n\left\|\overline{\bm{x}}_s - \bm{x}_s^{(i)}\right\|^2 + 3n\zeta^2 
\end{align*}

From the inequality above and \eqref{lemma:xavekey_4} we have
\begin{align*}
&\sum_{j=1}^n\sum_{s=1}^T\sum_{i=1}^n\mathbb E\left\|\bm{x}_{s+1}^{(i,j)} - \overline{\bm{x}}_{s+1}^{(j)}\right\|^2 \\
\leq & \left(\frac{\gamma^2n\sigma^2}{1-\sqrt{\beta}} + \sum_{j=1}^n\frac{3n\zeta^2}{1-\sqrt{\beta}}\right)\sum_{s=1}^T\sum_{r=1}^s\beta^{\frac{s-r}{2}} + \frac{6L^2\gamma^2}{1-\sqrt{\beta}}\sum_{s=1}^T\sum_{r=1}^s\sum_{i=1}^n\beta^{\frac{s-r}{2}}\left\|\overline{\bm{x}}_r- \bm{x}_r^{(i)}\right\|^2\\
\leq & \frac{\gamma^2n\sigma^2T}{(1-\sqrt{\beta})^2} + \frac{n\zeta^2T}{(1-\sqrt{\beta})^2} + \frac{6L^2\gamma^2}{1-\sqrt{\beta}}\sum_{i=1}^n\sum_{s=1}^T\sum_{r=1}^s\beta^{\frac{s-r}{2}}\left\|\overline{\bm{x}}_r- \bm{x}_r^{(i)}\right\|^2\\
= & \frac{\gamma^2n\sigma^2T}{(1-\sqrt{\beta})^2} + \frac{3n\zeta^2T}{(1-\sqrt{\beta})^2} + \frac{6L^2\gamma^2}{1-\sqrt{\beta}}\sum_{i=1}^n\sum_{r=1}^T\sum_{s=r}^T\beta^{\frac{s-r}{2}}\left\|\overline{\bm{x}}_r- \bm{x}_r^{(i)}\right\|^2\\
= & \frac{\gamma^2n\sigma^2T}{(1-\sqrt{\beta})^2} + \frac{3n\zeta^2T}{(1-\sqrt{\beta})^2} + \frac{6L^2\gamma^2}{1-\sqrt{\beta}}\sum_{i=1}^n\sum_{r=1}^T\sum_{s=0}^{T-r}\beta^{\frac{s}{2}}\left\|\overline{\bm{x}}_r- \bm{x}_r^{(i)}\right\|^2\\
\leq  & \frac{\gamma^2n\sigma^2T}{(1-\sqrt{\beta})^2} + \frac{3n\zeta^2T}{(1-\sqrt{\beta})^2} + \frac{6L^2\gamma^2}{(1-\sqrt{\beta})^2}\sum_{i=1}^n\sum_{r=1}^T\left\|\overline{\bm{x}}_r- \bm{x}_r^{(i)}\right\|^2.
\end{align*}
If $\gamma$ is small enough that satisfies $\left(1- \frac{6L^2\gamma^2}{(1-\sqrt{\beta})^2} \right) > 0$, then we have
\begin{align*}
\left(1- \frac{6L^2\gamma^2}{(1-\sqrt{\beta})^2} \right) \sum_{s=1}^T\sum_{i=1}^n\mathbb E\left\|\bm{x}_{s}^{(i)} - \overline{\bm{x}}_{s}\right\|^2 \leq & \frac{\gamma^2n\sigma^2T}{(1-\sqrt{\beta})^2} + \frac{3n\zeta^2T}{(1-\sqrt{\beta})^2}.
\end{align*}
Denote $C_1:=\left(1- \frac{6L^2\gamma^2}{(1-\sqrt{\beta})^2} \right)^{-1}$, then we have
\begin{align*}
\sum_{s=1}^T\sum_{i=1}^n\mathbb E\left\|\bm{x}_{s}^{(i)} - \overline{\bm{x}}_{s}\right\|^2 \leq & \frac{2\gamma^2n\sigma^2TC_1}{(1-\sqrt{\beta})^2} + \frac{6n\zeta^2TC_1}{(1-\sqrt{\beta})^2}.
\end{align*}

\end{proof}

\subsection{Proof to Theorem~\ref{theo:1}}

\begin{lemma}\label{L:deltax}
From the updating rule (\ref{eq: updatingrule}) and Assumption \ref{ass:global}, we have
\begin{align*}
\mathbb{E}_{t,P}\big[\left\|\Delta\overline{\bm{x}}_t\right\|^2\big] = &\frac{\alpha_2}{n}\sum_{j=1}^n \Tr\left( \left(V_t^{(j)}\right)^{\top}\left(I_n-A_n\right)V_t^{(j)} \right) +  \gamma^2 \left\| \overline{\bm{g}}(X_t;\Xi_t) \right\|^2,\\
\mathbb{E}_{t,P}[\Delta\overline{\bm{x}}_t]= &-\gamma \overline{\bm{g}}(X_t; \Xi_t).
\end{align*}
\begin{proof}
We begin with $\mathbb{E}_{t,P}\big[\left\|\Delta\overline{\bm{x}}_t\right\|^2\big]$: 
\begin{align*}
\mathbb{E}_{t,P}\big[\left\|\Delta\overline{\bm{x}}_t\right\|^2\big] =&  \sum\limits_{j=1}^n \mathbb{E}_{t,P}\Big[\left\|\Delta^{(j)}\overline{\bm{x}}_t\right\|^2\Big]\\
\xlongequal{(\ref{re: 2})}& \sum_{j=1}^n\mathbb{E}_{t,P}\bigg[\left\| \left(X_{t+1}^{(j)} - X_{t}^{(j)}\right)\frac{\bm{1}_n}{n}\right\|^2\bigg]\\
=& \sum_{j=1}^n\mathbb{E}_{t,P} \bigg[\left\| \left(V_t^{(j)}W_t^{(j)} - X_{t}^{(j)}\right)\frac{\bm{1}_n}{n}\right\|^2\bigg]\\
\xlongequal{(\ref{eq: T})}& \frac{1}{n}\sum_{j=1}^n\mathbb{E}_{t,P} \bigg[\Tr\left( \left(V_t^{(j)}W_t^{(j)} - X_{t}^{(j)}\right) A_n \left( \left(W_t^{(j)}\right)^{\top}\left(V_t^{(j)}\right)^{\top} - \left(X_{t}^{(j)}\right)^{\top}\right)   \right)\bigg]\\
= & \frac{1}{n}\sum_{j=1}^n \Tr\left( V_t^{(j)}\mathbb{E}_{t,P}\left[ W_t^{(j)} A_n \left(W_t^{(j)}\right)^{\top}\right] \left(V_t^{(j)}\right)^{\top} \right)\\
& - \frac{2}{n}\sum_{j=1}^n \Tr\left( X_{t}^{(j)} A_n\mathbb{E}_{t,P}\left[\left(W_t^{(j)}\right)^{\top}\right] \left(V_t^{(j)}\right)^{\top} \right) + \frac{1}{n}\sum_{j=1}^n \Tr\left( X_{t}^{(j)} A_n \left(X_{t}^{(j)}\right)^{\top}\right)\\
= & \frac{\alpha_2}{n}\sum_{j=1}^n \Tr\left( V_t^{(j)}\left(I_n-A_n\right)\left(V_t^{(j)}\right)^{\top} \right) + \frac{1}{n}\sum_{j=1}^n \Tr\left( V_t^{(j)} A_n \left(V_t^{(j)}\right)^{\top}  \right)\\ 
& - \frac{2}{n}\sum_{j=1}^n \Tr\left( X_{t}^{(j)} A_n \left(V_t^{(j)}\right)^{\top} \right) + \frac{1}{n}\sum_{j=1}^n \Tr\left( X_{t}^{(j)} A_n \left(X_{t}^{(j)}\right)^{\top}\right),\numberthis\label{lemma:deltax_1}
\end{align*}
where for the last two equations, we use Lemma (\ref{lem: EWAW}), Lemma(\ref{lem: EW}), and (\ref{re: 4}).
From (\ref{re: 5}), we can obtain the following equation: 
\begin{align*}
 V_t^{(j)}A_n \left(V_t^{(j)}\right)^{\top} = &X_t^{(j)} A_n \left(X_t^{(j)}\right)^{\top}  - \gamma G^{(j)}(X_t;\Xi_t)  A_n \left(X_t^{(j)}\right)^{\top} - \gamma X_t^{(j)}A_n\left(G^{(j)}(X_t;\Xi_t)\right)^{\top}\\
& + \gamma^2 G^{(j)}(X_t;\Xi_t) A_n \left( G^{(j)}(X_t;\Xi_t) \right)^{\top}\\
X_t^{(j)} A_n \left(V_t^{(j)}\right)^{\top} = &  X_t^{(j)} A_n \left(X_t^{(j)}\right)^{\top} - \gamma G^{(j)}(X_t;\Xi_t) A_n \left(X_t^{(j)}\right)^{\top} ,\numberthis\label{lemma:deltax_2}
\end{align*}
From the property of trace, we have:
\begin{equation} \label{lemma: deltax_3}
	\Tr \left(G^{(j)}(X_t;\Xi_t)  A_n \left(X_t^{(j)}\right)^{\top}\right) = \Tr \left(X_t^{(j)}A_n^{\top}\left(G^{(j)}(X_t;\Xi_t)\right)^{\top}\right) = \Tr \left(X_t^{(j)}A_n\left(G^{(j)}(X_t;\Xi_t)\right)^{\top}\right).
\end{equation}
Combing \eqref{lemma:deltax_1}, \eqref{lemma:deltax_2} and \eqref{lemma: deltax_3}, we have
\begin{align*}
\mathbb{E}_{t,P}\left\|\Delta\overline{\bm{x}}_t\right\|^2 = & \frac{\alpha_2}{n}\sum_{j=1}^n \Tr\left( V_t^{(j)}\left(I_n-A_n\right)\left(V_t^{(j)}\right)^{\top} \right) +  \frac{\gamma^2}{n}\sum_{j=1}^n \Tr\left( G^{(j)}(X_t;\Xi_t) A_n \left( G^{(j)}(X_t;\Xi_t) \right)^{\top}\right)\\
\xlongequal{(\ref{eq: T})} & \frac{\alpha_2}{n}\sum_{j=1}^n \Tr\left( V_t^{(j)}\left(I_n-A_n\right)\left(V_t^{(j)}\right)^{\top} \right) +  \gamma^2\sum_{j=1}^n \left\| G^{(j)}(X_t;\Xi_t)\frac{\bm{1}_n}{n} \right\|^2\\
\xlongequal{(\ref{re: 9})} & \frac{\alpha_2}{n}\sum_{j=1}^n \Tr\left( V_t^{(j)}\left(I_n-A_n\right)\left(V_t^{(j)}\right)^{\top} \right) +  \gamma^2 \left\| \overline{\bm{g}}(X_t;\Xi_t) \right\|^2.
\end{align*}
For $\mathbb{E}_{t,P}[\Delta\overline{\bm{x}}_t]$, we first compute $\mathbb{E}_{t,P}[\Delta^{(j)}\overline{\bm{x}}_t], (j\in [n])$.
\begin{align*}
\mathbb{E}_{t,P}[\Delta^{(j)}\overline{\bm{x}}_t] &= \mathbb{E}_{t,P}[\overline{\bm{x}}_{t+1}^{(j)}] - \mathbb{E}_{t,P}[\overline{\bm{x}}_t^{(j)}]\\
&=\mathbb{E}_{t,P}\left[X_t^{(j)}\right]\frac{\mathbf{1}}{n} - \overline{\bm{x}}_t^{(j)}\\
&=V_t^{(j)}\mathbb{E}_{t,P}\left[W_t^{(j)}\right]\frac{\bm{1}}{n} - \overline{\bm{x}}_t^{(j)}\\
&\xlongequal{\text{Lemma (\ref{lem: EW})}} V_t^{(j)}\left(\alpha_1I_n +(1-\alpha_1)A_n\right)\frac{\bm{1}}{n} - \overline{\bm{x}}_t^{(j)}\\
&=  V_t^{(j)}\frac{\bm{1}}{n} - \overline{\bm{x}}_t^{(j)}\\
&= \overline{\bm{v}}_t^{(j)} - \overline{\bm{x}}_t^{(j)}\\
&=  -\gamma \overline{\bm{g}}^{(j)}(X_t;\Xi_t),
\end{align*}
which immediately leads to $\mathbb{E}_{t,P}\left[\Delta\overline{\bm{x}}_t\right] = -\gamma \overline{\bm{g}}(X_t; \Xi_t)$.
\end{proof}

\end{lemma}

\begin{lemma}\label{L:key}
From the updating rule (\ref{eq: updatingrule}) and Assumption \ref{ass:global}, we have
\begin{align*}
\sum_{j=1}^n\mathbb{E}_{t,G}\left[\Tr\left( V_t^{(j)}\left(I_n-A_n\right)\left(V_t^{(j)}\right)^{\top} \right)\right]
\leq & (2+12\gamma^2L)\sum_{i=1}^n\left\|\left(\bm{x}_t^{(i)} - \overline{\bm{x}}_t\right)\right\|^2 + 6n\gamma^2\zeta^2 + 2n\gamma^2\sigma^2.\end{align*}

\begin{proof}
\begin{align*}
\Tr\left( V_t^{(j)}\left(I_n-A_n\right)\left(V_t^{(j)}\right)^{\top} \right) = & \Tr\left( V_t^{(j)}\left(V_t^{(j)}\right)^{\top}\right) - \Tr\left( V_t^{(j)}A_n\left(V_t^{(j)}\right)^{\top} \right)\\
\xlongequal{(\ref{eq: T}} & \left\|V_t^{(j)}\right\|^2_F - n\left\|V_t^{(j)}\frac{\bm{1}_n}{n}\right\|^2\\
= & \sum_{i=1}^n\left(\left\|\bm{v}_t^{(i,j)}\right\|^2 - \left\|\overline{\bm{v}}_t^{(j)}\right\|^2\right)\\
= & \sum_{i=1}^n\left\|\bm{v}_t^{(i,j)} - \overline{\bm{v}}_t^{(j)}\right\|^2,\numberthis\label{lemma:key_1}
\end{align*}
the last equation above is because
\begin{align} \label{eq: a}
\sum_{i=1}^n \|\bm{a}_i\|^2 - \left\|\sum_{i=1}^n\frac{\bm{a}_i}{n}\right\|^2 = \sum_{i=1}^n\left\|\bm{a}_i-\sum_{k=1}^n\frac{\bm{a}_i}{n}\right\|^2.
\end{align}

Since
\begin{align*}
\overline{\bm{v}}_t^{(j)} = & \overline{\bm{x}}_t^{(j)} - \gamma \overline{\bm{g}}^{(j)}(X_t;\Xi_t)\\
\bm{v}_t^{(i,j)} - \overline{\bm{v}}_t^{(j)} = & \left(\bm{x}_t^{(i,j)} - \overline{\bm{x}}_t^{(j)}\right) - \gamma\left(\bm{g}^{(i,j)}(\bm{x}_t^{(i)};\bm{\xi}_t^{(i)}) - \overline{\bm{g}}^{(j)}(X_t;\Xi_t)\right),
\end{align*}
we have the following: 
\begin{align*}
\sum_{i=1}^n\left\|\bm{v}_t^{(i,j)} - \overline{\bm{v}}_t^{(j)}\right\|^2 = & \sum_{i=1}^n\left\|\left(\bm{x}_t^{(i,j)} - \overline{\bm{x}}_t^{(j)}\right) - \gamma\left(\bm{g}^{(i,j)}(\bm{x}_t^{(i)};\bm{\xi}_t^{(i)}) - \overline{\bm{g}}^{(j)}(X_t;\Xi_t)\right)\right\|^2\\
\leq & 2\sum_{i=1}^n\left\|\left(\bm{x}_t^{(i,j)} - \overline{\bm{x}}_t^{(j)}\right)\right\|^2 + 2\gamma^2\sum_{i=1}^n\left\|\bm{g}^{(i,j)}(\bm{x}_t^{(i)};\bm{\xi}_t^{(i)}) - \overline{\bm{g}}^{(j)}(X_t;\Xi_t)\right\|^2,\numberthis\label{lemma:key_2}
\end{align*}
where the last inequality comes from $\left\|\bm{a}+\bm{b}\right\|^2\le 2\left\|\bm{a}\right\|^2 + 2\left\|\bm{b}\right\|^2$.

For $\left\|\bm{g}^{(i,j)}(\bm{x}_t^{(i)};\bm{\xi}_t^{(i)}) - \overline{\bm{g}}^{(j)}(X_t;\Xi_t)\right\|^2$, we have
\begin{align*}
&\sum_{i=1}^n\mathbb E_{t,G}\left\|\bm{g}^{(i,j)}(\bm{x}_t^{(i)};\bm{\xi}_t^{(i)}) - \overline{\bm{g}}^{(j)}(X_t;\Xi_t)\right\|^2\\
\xlongequal{(\ref{eq: a})} & \sum_{i=1}^n\mathbb E_{t,G}\left\|\bm{g}^{(i,j)}(\bm{x}_t^{(i)};\bm{\xi}_t^{(i)}) \right\|^2 - n\mathbb E_{t,G}\left\|\overline{\bm{g}}^{(j)}(X_t;\Xi_t)\right\|^2\\
 = & \sum_{i=1}^n\mathbb E_{t,G}\left\|\left(\bm{g}^{(i,j)}(\bm{x}_t^{(i)};\bm{\xi}_t^{(i)}) -\nabla^{(j)} f_i(\bm{x}_t^{(i)})\right) + \nabla^{(j)} f_i(\bm{x}_t^{(i)})\right\|^2\\
 & - n\mathbb E_{t,G}\left\|\left(\overline{\bm{g}}^{(j)}(X_t;\Xi_t) - \overline{\nabla}^{(j)}f(X_t)\right) + \overline{\nabla}^{(j)}f(X_t) \right\|^2\\
=& \sum_{i=1}^n\mathbb E_{t,G}\left\|\bm{g}^{(i,j)}(\bm{x}_t^{(i)};\bm{\xi}_t^{(i)}) -\nabla^{(j)} f_i(\bm{x}_t^{(i)})\right\|^2 + \sum_{i=1}^n\left\|\nabla^{(j)} f_i(\bm{x}_t^{(i)})\right\|^2 
  - n\mathbb E_{t,G}\left\|\overline{\bm{g}}^{(j)}(X_t;\Xi_t) - \overline{\nabla}^{(j)}f(X_t)\right\|^2\\
  & - n\left\| \overline{\nabla}^{(j)}f(X_t) \right\|^2
  + 2\sum_{i=1}^n\mathbb E_{t,G}\left[\left\langle\bm{g}^{(i,j)}(\bm{x}_t^{(i)};\bm{\xi}_t^{(i)}) -\nabla^{(j)} f_i(\bm{x}_t^{(i)}),\nabla^{(j)} f_i(\bm{x}_t^{(i)})\right\rangle\right]\\
 & - 2n\mathbb E_{t,G}\left[\left\langle\overline{\bm{g}}^{(j)}(X_t;\Xi_t) - \overline{\nabla}^{(j)}f(X_t),\overline{\nabla}^{(j)}f(X_t)\right\rangle\right]\\
 = & \sum_{i=1}^n\mathbb E_{t,G}\left\|\bm{g}^{(i,j)}(\bm{x}_t^{(i)};\bm{\xi}_t^{(i)}) -\nabla^{(j)} f_i(\bm{x}_t^{(i)})\right\|^2 + \sum_{i=1}^n\left\|\nabla^{(j)} f_i(\bm{x}_t^{(i)})\right\|^2 
  - n\mathbb E_{t,G}\left\|\overline{\bm{g}}^{(j)}(X_t;\Xi_t) - \overline{\nabla}^{(j)}f(X_t)\right\|^2\\
& - n\left\| \overline{\nabla}^{(j)}f(X_t) \right\|^2\\
\leq & \sum_{i=1}^n\mathbb E_{t,G}\left\|\bm{g}^{(i,j)}(\bm{x}_t^{(i)};\bm{\xi}_t^{(i)}) -\nabla^{(j)} f_i(\bm{x}_t^{(i)})\right\|^2 + \sum_{i=1}^n\left\|\nabla^{(j)} f_i(\bm{x}_t^{(i)})\right\|^2  - n\left\| \overline{\nabla}^{(j)}f(X_t) \right\|^2.\numberthis\label{lemma:key_3}
\end{align*}
For $\sum_{i=1}^n\left\|\nabla^{(j)} f_i(\bm{x}_t^{(i)})\right\|^2  - n\left\| \overline{\nabla}^{(j)}f(X_t) \right\|^2$, we have
\begin{align*}
&\sum_{i=1}^n\left\|\nabla^{(j)} f_i(\bm{x}_t^{(i)})\right\|^2  - n\left\| \overline{\nabla}^{(j)}f(X_t) \right\|^2\\
\xlongequal{(\ref{eq: a})} & \sum_{i=1}^n\left\|\nabla^{(j)} f_i(\bm{x}_t^{(i)}) - \overline{\nabla}^{(j)}f(X_t) \right\|^2\\
= &  \sum_{i=1}^n\left\|\left(\nabla^{(j)} f_i(\bm{x}_t^{(i)}) - \nabla^{(j)} f_i(\overline{\bm{x}}_t)\right) - \left(\overline{\nabla}^{(j)}f(X_t) - \nabla^{(j)} f(\overline{\bm{x}}_t)\right) + \left(\nabla^{(j)} f_i(\overline{\bm{x}}_t) - \nabla^{(j)} f(\overline{\bm{x}}_t)\right) \right\|^2\\
\leq &  3\sum_{i=1}^n\left\|\nabla^{(j)} f_i(\bm{x}_t^{(i)}) - \nabla^{(j)} f_i(\overline{\bm{x}}_t)\right\|^2 + 3\sum_{i=1}^n\left\| \overline{\nabla}^{(j)}f(X_t) - \nabla^{(j)} f(\overline{\bm{x}}_t)\right\|^2 \\
& + 3\sum_{i=1}^n\left\| \nabla^{(j)} f_i(\overline{\bm{x}}_t) - \nabla^{(j)} f(\overline{\bm{x}}_t) \right\|^2 \text{        }\left(\text{due to } \left\|\bm{a} + \bm{b} + \bm{c}\right\|^2\le 3\left\|\bm{a}\right\|^2 + 3\left\|\bm{}\right\|^2 + 3\left\|\bm{c}\right\|^2\right)\\
\leq &  3L^2\sum_{i=1}^n\left\|\bm{x}_t^{(i,j)} - \overline{\bm{x}}_t^{(j)}\right\|^2  + 3n\left\| \overline{\nabla}^{(j)}f(X_t) - \nabla^{(j)} f(\overline{\bm{x}}_t)\right\|^2+ 3\sum_{i=1}^n\left\| \nabla^{(j)} f_i(\overline{\bm{x}}_t) - \nabla^{(j)} f(\overline{\bm{x}}_t) \right\|^2\\
= &  3L^2\sum_{i=1}^n\left\|\bm{x}_t^{(i,j)} - \overline{\bm{x}}_t^{(j)}\right\|^2  + \frac{3}{n}\left\|\sum_{k=1}^n \left( \nabla^{(j)} f_k(\bm{x}_t^{(k)}) - \nabla^{(j)} f_i(\overline{\bm{x}}_t) \right)\right\|^2+ 3\sum_{i=1}^n\left\| \nabla^{(j)} f_i(\overline{\bm{x}}_t) - \nabla^{(j)} f(\overline{\bm{x}}_t) \right\|^2\\
\leq &  3L^2\sum_{i=1}^n\left\|\bm{x}_t^{(i,j)} - \overline{\bm{x}}_t^{(j)}\right\|^2  + 3\sum_{k=1}^n \left\| \nabla^{(j)} f_k(\bm{x}_t^{(k)}) - \nabla^{(j)} f_k(\overline{\bm{x}}_t)\right\|^2+ 3\sum_{i=1}^n\left\| \nabla^{(j)} f_i(\overline{\bm{x}}_t) - \nabla^{(j)} f(\overline{\bm{x}}_t) \right\|^2\\
\leq &  3L^2\sum_{i=1}^n\left\|\bm{x}_t^{(i,j)} - \overline{\bm{x}}_t^{(j)}\right\|^2  + 3L^2\sum_{k=1}^n\left\|\bm{x}_t^{(k,j)} - \overline{\bm{x}}_t^{(j)}\right\|^2+ 3\sum_{i=1}^n\left\| \nabla^{(j)} f_i(\overline{\bm{x}}_t) - \nabla^{(j)} f(\overline{\bm{x}}_t) \right\|^2\\
\leq &  6L^2\sum_{i=1}^n\left\|\bm{x}_t^{(i,j)} - \overline{\bm{x}}_t^{(j)}\right\|^2 + 3\sum_{i=1}^n\left\| \nabla^{(j)} f_i(\overline{\bm{x}}_t) - \nabla^{(j)} f(\overline{\bm{x}}_t) \right\|^2.
\end{align*}
Taking the above inequality into \eqref{lemma:key_3}, we get
\begin{align*}
&\sum_{i=1}^n\mathbb E_{t,G}\left\|\bm{g}^{(i,j)}(\bm{x}_t^{(i)};\bm{\xi}_t^{(i)}) - \overline{\bm{g}}^{(j)}(X_t;\Xi_t)\right\|^2\\ 
\leq & \sum_{i=1}^n\mathbb E_{t,G}\left\|\bm{g}^{(i,j)}(\bm{x}_t^{(i)};\bm{\xi}_t^{(i)}) -\nabla^{(j)} f_i(\bm{x}_t^{(i)})\right\|^2  + 3\sum_{i=1}^n\left\| \nabla^{(j)} f_i(\overline{\bm{x}}_t) - \nabla^{(j)} f(\overline{\bm{x}}_t) \right\|^2\\
&+ 6L^2\sum_{i=1}^n\left\|\bm{x}_t^{(i,j)} - \overline{\bm{x}}_t^{(j)}\right\|^2.\numberthis\label{lemma:key_4}
\end{align*}

Combinig \eqref{lemma:key_2} and \eqref{lemma:key_4} together we have
\begin{align*}
\mathbb E_{t,G}\sum_{i=1}^n\left\|\bm{v}_t^{(i,j)} - \overline{\bm{v}}_t^{(j)}\right\|^2 \leq & (2+12L^2\gamma^2)\sum_{i=1}^n\left\|\bm{x}_t^{(i,j)} - \overline{\bm{x}}_t^{(j)}\right\|^2 + 6\gamma^2\sum_{i=1}^n\left\| \nabla^{(j)} f_i(\overline{\bm{x}}_t) - \nabla^{(j)} f(\overline{\bm{x}}_t) \right\|^2\\
& + 2\gamma^2\sum_{i=1}^n\mathbb E_{t,G}\left\|\bm{g}^{(i,j)}(\bm{x}_t^{(i)};\bm{\xi}_t^{(i)}) -\nabla^{(j)} f_i(\bm{x}_t^{(i)})\right\|^2
\end{align*}
Summing $j$ from $1$ to $n$, we obtain the following: 
\begin{align*}
\mathbb E_{t,G}\sum_{i=1}^n\left\|\bm{v}_t^{(i)} - \overline{\bm{v}}_t\right\|^2 \leq & (2+12L^2\gamma^2)\sum_{i=1}^n\left\|\bm{x}_t^{(i)} - \overline{\bm{x}}_t\right\|^2 + 6\gamma^2\sum_{i=1}^n\left\| \nabla f_{i}(\overline{\bm{x}}_t) - \nabla f(\overline{\bm{x}}_t) \right\|^2\\
& + 2\gamma^2\sum_{i=1}^n\mathbb E_{t,G}\left\|\bm{g}^{(i)}(\bm{x}^{(i)}_t;\bm{\xi}^{(i)}_t) -\nabla f_i(\bm{x}_t^{(i)})\right\|^2\\
\leq & (2+12L^2\gamma^2)\sum_{i=1}^n\left\|\left(\bm{x}_t^{(i)} - \overline{\bm{x}}_t\right)\right\|^2 + 6n\gamma^2\zeta^2
 + 2n\gamma^2\sigma^2.\numberthis\label{lemma:key_5}
\end{align*}
From \eqref{lemma:key_1} and \eqref{lemma:key_5}, we have
\begin{align*}
\sum_{j=1}^n\mathbb{E}_{t,G}\Tr\left( V_t^{(j)}\left(I_n-A_n\right)\left(V_t^{(j)}\right)^{\top} \right)
\leq & (2+12\gamma^2L)\sum_{i=1}^n\left\|\left(\bm{x}_t^{(i)} - \overline{\bm{x}}_t\right)\right\|^2 + 6n\gamma^2\zeta^2 + 2n\gamma^2\sigma^2.
\end{align*}

\end{proof}
\end{lemma}

\begin{proof} [\textbf{Proof to Theorem~\ref{theo:1}}]
From Lemma \ref{lem: EW} and Lemma \ref{lem: EWAW}, we have
\begin{align*}
\mathbb{E}_{t,P}(W_t^{(j)}) =& \alpha_1I_n + (1-\alpha_1)A_n\\
\mathbb{E}_{t,P}\left(W_t^{(j)}A_n\left(W_t^{(j)}\right)^{\top}\right) =& \alpha_2 I_n + (1-\alpha_2)A_n
\end{align*}

From the updating rule (\ref{eq: updatingrule}) and L-Lipschitz of $f$, we have
\begin{align*}
\mathbb{E}_{t,P}f(\overline{X}_{t+1}) \leq & f(\overline{X}_{t}) + \mathbb{E}_{t,P}\langle \nabla f(\overline{X}_{t}), \Delta\overline{\bm{x}}_t \rangle + \mathbb{E}_{t,P}\frac{L}{2}\left\|\overline{\bm{x}}_t\right\|^2\\
\xlongequal{\text{Lemma \ref{L:deltax}}}& f(\overline{X}_{t}) - \gamma\langle \nabla f(\overline{X}_{t}), \gamma \overline{\bm{g}}(X_t; \Xi_t) \rangle + \mathbb{E}_{t,P}\frac{L}{2}\left\|\overline{\bm{x}}_t\right\|^2\\
\xlongequal{\text{Lemma \ref{L:deltax}}} & f(\overline{X}_{t}) - \gamma\langle \nabla f(\overline{X}_{t}), \gamma \overline{\bm{g}}(X_t; \Xi_t) \rangle + \frac{\alpha_2 L}{2n}\sum_{j=1}^n \Tr\left( V_t^{(j)}\left(I_n-A_n\right)\left(V_t^{(j)}\right)^{\top} \right)\\
& +  \sum_{j=1}^n\frac{\gamma^2L}{2} \mathbb E_{t,G}\left\| \overline{\bm{g}}^{(j)}(X_t;\Xi_t) \right\|^2.
\end{align*}
So
\begin{align*}
\mathbb{E}_{t}f(\overline{X}_{t+1}) \leq & f(\overline{X}_{t}) -\gamma \langle \nabla f(\overline{X}_{t}), \overline{\nabla} f(X_t) \rangle + \frac{\alpha_2 L}{2n}\sum_{j=1}^n \mathbb E_{t,G}\Tr\left( V_t^{(j)}\left(I_n-A_n\right)\left(V_t^{(j)}\right)^{\top} \right)\\
& +  \frac{\gamma^2L}{2} \sum_{j=1}^n\mathbb E_{t,G}\left\| \overline{\bm{g}}^{(j)}(X_t;\Xi_t) \right\|^2.\numberthis\label{main:1_1}
\end{align*}

Since
\begin{align*}
\mathbb E_{t,G}\left\|\overline{\bm{g}}^{(j)}(X_t;\Xi_t)\right\|^2 =& \mathbb E_{t,G}\left\|\left( \overline{\bm{g}}^{(j)}(X_t;\Xi_t)- \overline{\nabla}^{(j)}f(X_t)\right) + \overline{\nabla}^{(j)}f(X_t) \right\|^2\\
= & \mathbb E_{t,G}\left\|\overline{\bm{g}}^{(j)}(X_t;\Xi_t)- \overline{\nabla}^{(j)}f(X_t)\right\|^2+ \mathbb E \left\|\overline{\nabla}^{(j)}f(X_t) \right\|^2\\
& + 2E_{t,G}\left\langle \overline{\bm{g}}^{(j)}(X_t;\Xi_t)- \overline{\nabla}^{(j)}f(X_t), \overline{\nabla}^{(j)}f(X_t) \right\rangle\\
= & \mathbb E_{t,G}\left\|\overline{\bm{g}}^{(j)}(X_t;\Xi_t)- \overline{\nabla}^{(j)}f(X_t)\right\|^2+ \left\|\overline{\nabla}^{(j)}f(X_t) \right\|^2,
\end{align*}
and
\begin{align*}
&\sum_{j=1}^n\mathbb E_{t,G}\left\|\overline{\bm{g}}^{(j)}(X_t;\Xi_t) - \overline{\nabla}^{(j)}f(X_t)\right\|^2\\
 = & \frac{1}{n^2}\sum_{j=1}^n\mathbb E_{t,G}\left\| \sum_{i=1}^n \left(\bm{g}^{(i,j)}(\bm{x}_t^{(i)};\bm{\xi}_t^{(i)}) - \nabla^{(j)} f_i(\bm{x}_t^{(i)})\right)\right\|^2\\
= & \frac{1}{n^2}\sum_{j=1}^n\sum_{i=1}^n \mathbb E_{t,G}\left\| \bm{g}^{(i,j)}(\bm{x}_t^{(i)};\bm{\xi}_t^{(i)}) - \nabla^{(j)} f_i(\bm{x}_t^{(i)})\right\|^2\\
& + \frac{1}{n^2}\sum_{j=1}^n\mathbb E_{t,G}\sum_{i\not=i'}\left\langle \bm{g}^{(i,j)}(\bm{x}_t^{(i)};\bm{\xi}_t^{(i)}) - \nabla^{(j)} f_i(\bm{x}_t^{(i)}), \bm{g}^{(i',j)}(\bm{x}_t^{(i')};\bm{\xi}_t^{(i')}) - \nabla^{(j)} f_{i'}(\bm{x}_t^{(i')}) \right\rangle\\
=&\frac{1}{n^2}\sum_{j=1}^n\sum_{i=1}^n \mathbb E_{t,G}\left\| \bm{g}^{(i,j)}(\bm{x}_t^{(i)};\bm{\xi}_t^{(i)}) - \nabla^{(j)} f_i(\bm{x}_t^{(i)})\right\|^2\\
=&\frac{1}{n^2}\sum_{i=1}^n \mathbb E_{t,G}\left\| \bm{g}^{(i)}(\bm{x}_t^{(i)};\bm{\xi}_t^{(i)}) - \nabla f_i(\bm{x}_t^{(i)})\right\|^2\\
\leq & \frac{1}{n^2}\sum\limits_{i=1}^n \sigma^2\\
= & \frac{\sigma^2}{n}
\end{align*}
then we have 
\begin{align*}
\mathbb E_{t,G}\sum_{j=1}^n\left\|\overline{\bm{g}}^{(j)}(X_t;\Xi_t)\right\|^2 \leq & \frac{\sigma^2}{n} + \sum_{j=1}^n \left\|\overline{\nabla}^{(j)}f(X_t) \right\|^2 \numberthis\label{main:1_2}
\end{align*}
Combining \eqref{main:1_1} and \eqref{main:1_2}, we have
\begin{align*}
\mathbb{E}_{t}f(\overline{X}_{t+1})\leq & f(\overline{X}_{t}) - \gamma\langle \nabla f(\overline{X}_{t}), \overline{\nabla} f(X_t) \rangle + \frac{\alpha_2 L}{2n}\sum_{j=1}^n \mathbb E_{t,G}\Tr\left( V_t^{(j)}\left(I_n-A_n\right)\left(V_t^{(j)}\right)^{\top} \right)\\
& +  \frac{\gamma^2L\sigma^2}{2n} + \frac{\gamma^2L}{2}\sum_{j=1}^n \left\|\overline{\nabla}^{(j)}f(X_t) \right\|^2\\
\leq &f(\overline{X}_{t}) - \gamma\langle \nabla f(\overline{X}_{t}), \overline{\nabla} f(X_t) \rangle  +  \frac{\gamma^2L\sigma^2}{2n} + \frac{\gamma^2L}{2}\sum_{j=1}^n \left\|\overline{\nabla}^{(j)}f(X_t) \right\|^2\\
& + \frac{\alpha_2 L(2+ 12L^2\gamma^2)}{2n}\sum_{i=1}^n\sum_{j=1}^n\left\|\left(\bm{x}_t^{(i,j)} - \overline{\bm{x}}_t^{(j)}\right)\right\|^2 + 2\alpha_2 L^2\gamma^2\sigma^2 + 6\alpha_2 L\zeta^2\gamma^2\quad \text{(due to Lemma~\ref{L:key} )}\\
= & f(\overline{X}_{t}) - \gamma\langle \nabla f(\overline{X}_{t}), \overline{\nabla} f(X_t) \rangle  +  \frac{\gamma^2L\sigma^2}{2n} + \frac{\gamma^2L}{2}\left\|\overline{\nabla}f(X_t) \right\|^2\\
& + \frac{\alpha_2 L(2+ 12L^2\gamma^2)}{2n}\sum_{i=1}^n\left\|\left(\bm{x}_t^{(i)} - \overline{\bm{x}}_t\right)\right\|^2 + 2\alpha_2 L\sigma^2\gamma^2 + 6\alpha_2 L\zeta^2\gamma^2\\
= & f(\overline{X}_{t}) -  \frac{\gamma}{2}\left\|\nabla f(\overline{X}_{t})\right\|^2 -\frac{\gamma}{2}\left\|\overline{\nabla} f(X_t)\right\|^2 + \frac{\gamma}{2}\left\|\nabla f(\overline{X}_{t})-\overline{\nabla} f(X_t)\right\|^2 +  \frac{\gamma^2L\sigma^2}{2n}\\
& + \frac{\gamma^2L}{2}\left\|\overline{\nabla}f(X_t) \right\|^2 + \frac{\alpha_2 L(2+ 12L^2\gamma^2)}{2n}\sum_{i=1}^n\left\|\left(\bm{x}_t^{(i)} - \overline{\bm{x}}_t\right)\right\|^2 + 2\alpha_2 L\sigma^2\gamma^2 + 6\alpha_2 L\zeta^2\gamma^2. \numberthis \label{main:2_1}
\end{align*}
Since 
\begin{align*}
\left\|\nabla f(\overline{X}_{t})-\overline{\nabla} f(X_t)\right\|^2 =& \frac{1}{n^2}{\left\|
	\sum_{i=1}^n\left(\nabla f_i(\overline{X}_t)  -  \nabla f_i(\bm{x}_t^{(i)})\right)\right\|^2}\\
\leq & \frac{1}{n}\sum_{i=1}^n\left\|\nabla f_i(\overline{X}_t)  -  \nabla f_i(\bm{x}_t^{(i)})\right\|^2\\
\leq & \frac{L^2}{n}\sum_{i=1}^n\left\|\overline{X}_t-\bm{x}_t^{(i)}\right\|^2.
\end{align*}
So \eqref{main:2_1} becomes
\begin{align*}
\mathbb{E}_{t}f(\overline{X}_{t+1})
\leq & f(\overline{X}_{t}) -  \frac{\gamma}{2}\left\|\nabla f(\overline{X}_{t})\right\|^2 -\frac{\gamma}{2}\left\|\overline{\nabla} f(X_t)\right\|^2 +  \frac{\gamma^2L\sigma^2}{2n}\\
& + \frac{\gamma^2L}{2}\left\|\overline{\nabla}f(X_t) \right\|^2 + \frac{\alpha_2 L(2+ 12L^2\gamma^2) + L^2\gamma}{2n}\sum_{i=1}^n\left\|\left(\bm{x}_t^{(i)} - \overline{\bm{x}}_t\right)\right\|^2 + 2\alpha_2 L\sigma^2\gamma^2 + 6\alpha_2 L\zeta^2\gamma^2\\
= & f(\overline{X}_{t}) -  \frac{\gamma}{2}\left\|\nabla f(\overline{X}_{t})\right\|^2 -\frac{\gamma(1-L\gamma)}{2}\left\|\overline{\nabla} f(X_t)\right\|^2 +  \frac{\gamma^2L\sigma^2}{2n} + 2\alpha_2 L\sigma^2\gamma^2 + 6\alpha_2 L\zeta^2\gamma^2\\
&  + \frac{\alpha_2 L(2+ 12L^2\gamma^2) + L^2\gamma}{2n}\sum_{i=1}^n\left\|\left(\bm{x}_t^{(i)} - \overline{\bm{x}}_t\right)\right\|^2 .
\end{align*}
Taking the expectation over the whole history, the inequality above becomes
\begin{align*}
\mathbb{E}f(\overline{X}_{t+1})
\leq & \mathbb{E}f(\overline{X}_{t}) -  \frac{\gamma}{2}\mathbb{E}\left\|\nabla f(\overline{X}_{t})\right\|^2 -\frac{\gamma(1-L\gamma)}{2}\mathbb{E}\left\|\overline{\nabla} f(X_t)\right\|^2 +  \frac{\gamma^2L\sigma^2}{2n} + 2\alpha_2 L\sigma^2\gamma^2 + 6\alpha_2 L\zeta^2\gamma^2\\
&  + \frac{\alpha_2 L(2+ 12L^2\gamma^2) + L^2\gamma}{2n}\sum_{i=1}^n\mathbb{E}\left\|\left(\bm{x}_t^{(i)} - \overline{\bm{x}}_t\right)\right\|^2 ,\\
\end{align*}
which implies
\begin{align*}
\frac{\gamma}{2}\mathbb{E}\left\|\nabla f(\overline{X}_{t})\right\|^2 + \frac{\gamma(1-L\gamma)}{2}\mathbb{E}\left\|\overline{\nabla} f(X_t)\right\|^2
\leq & \mathbb{E}f(\overline{X}_{t}) - \mathbb{E}f(\overline{X}_{t+1}) + \frac{\gamma^2L\sigma^2}{2n} + 2\alpha_2 L\sigma^2\gamma^2 + 6\alpha_2 L\zeta^2\gamma^2\\
& + \frac{\alpha_2 L(2+ 12L^2\gamma^2) + L^2\gamma}{2n}\sum_{i=1}^n\mathbb{E}\left\|\left(\bm{x}_t^{(i)} - \overline{\bm{x}}_t\right)\right\|^2.\numberthis\label{main:2_2}
\end{align*}
Summing up both sides of \eqref{main:2_2}, it becomes
\begin{align*}
&\sum_{t=1}^T\left( \mathbb{E}\left\|\nabla f(\overline{X}_{t})\right\|^2 + (1-L\gamma)\mathbb{E}\left\|\overline{\nabla} f(X_t)\right\|^2 \right)\\
\leq & \frac{2(f(\overline{\bm{x}}_0) - \mathbb Ef(\overline{\bm{x}}_{T+1}))}{\gamma}  + \frac{\gamma L\sigma^2 T}{n} + 4\alpha_2 L\sigma^2\gamma T + 12\alpha_2 L\zeta^2\gamma T\\
& + \frac{\alpha_2 L(2+ 12L^2\gamma^2) + L^2\gamma}{2n\gamma}\sum_{t=1}^T\sum_{i=1}^n\mathbb{E}\left\|\left(\bm{x}_t^{(i)} - \overline{\bm{x}}_t\right)\right\|^2.\numberthis\label{main:2_3}
\end{align*}
According to Lemma~\ref{L:xavekey}, we have
\begin{align*}
\sum_{t=1}^T\sum_{i=1}^n\mathbb E\left\|\bm{x}_{t}^{(i)} - \overline{\bm{x}}_{t}\right\|^2 \leq & \frac{2\gamma^2n\sigma^2TC_1}{(1-\sqrt{\beta})^2} + \frac{6n\zeta^2TC_1}{(1-\sqrt{\beta})^2},
\end{align*}
where $C_1=\left(1- \frac{6L^2\gamma^2}{(1-\sqrt{\beta})^2} \right)^{-1}$. Combing the inequality above with \eqref{main:2_3} we get
\begin{align*}
&\frac{1}{T}\sum_{t=1}^T\left( \mathbb{E}\left\|\nabla f(\overline{X}_{t})\right\|^2 + (1-L\gamma)\mathbb{E}\left\|\overline{\nabla} f(X_t)\right\|^2 \right)\\
\leq & \frac{2f(\bm{0}) -2f(\bm{x}^*)}{\gamma T}  + \frac{\gamma L\sigma^2 }{n} + 4\alpha_2 L\gamma(\sigma^2 + 3\zeta^2)\\
& + \frac{\left(\alpha_2 L\gamma(2+ 12L^2\gamma^2) + L^2\gamma^2\right)\sigma^2 C_1}{(1-\sqrt{\beta})^2} + \frac{3\left(\alpha_2 L\gamma(2+ 12L^2\gamma^2) + L^2\gamma^2\right)\zeta^2 C_1}{(1-\sqrt{\beta})^2}.
\end{align*}
\end{proof}

\section{Proof to Corollary~\ref{coro1}}
\begin{proof} [\textbf{Proof to Corollary~\ref{coro1}}]
Setting 
\begin{align*}
\gamma = \frac{1-\sqrt{\beta}}{6L + 3(\sigma+\zeta)\sqrt{\alpha_2 T} + \frac{\sigma\sqrt{T}}{\sqrt{n}}},
\end{align*}
then we have
\begin{align*}
1-L\gamma \geq & 0\\
C_1 \leq & 2\\
2 + 12 L^2\gamma^2 \leq & 4
\end{align*}
So \eqref{theo1eq} becomes
\begin{align*}
\frac{1}{T}\sum_{t=1}^T\mathbb{E}\left\|\nabla f(\overline{X}_{t})\right\|^2 
\leq & \frac{(2f(\bm{0}) -2f(\bm{x}^*) + L)\sigma}{\sqrt{nT}(1-\sqrt{\beta})} + \frac{(2f(\bm{0}) -2f(\bm{x}^*) + L)(\sigma+\zeta)}{1-\sqrt{\beta}}\sqrt{\frac{\alpha_2}{T}}\\
& + \frac{(2f(\bm{0}) -2f(\bm{x}^*))L}{T} + \frac{L^2(\sigma^2 + \zeta^2)}{(\frac{T}{n} + \alpha_2 T )\sigma^2 + \alpha_2 T \zeta^2},\\
\frac{1}{T}\sum_{t=1}^T\mathbb{E}\left\|\nabla f(\overline{X}_{t})\right\|^2 
\lesssim & \frac{\sigma + \zeta}{(1-\sqrt{\beta})\sqrt{nT}} + \frac{\sigma + \zeta}{(1-\sqrt{\beta})}\sqrt{\frac{\alpha_2}{T}}
 + \frac{1}{T} + \frac{n(\sigma^2 + \zeta^2)}{(1 + n\alpha_2  )\sigma^2 T + n\alpha_2 T \zeta^2}.
\end{align*}

\end{proof}

\section{Properties of Weighted Matrix $W_t^{(j)}$}\label{secD}
In this section, we will give three properties of $W_t^{(j)}$, described by Lemma \ref{lem: EW}, Lemma \ref{lem: EWW} and Lemma \ref{lem: EWAW}.

Throughout this section, we will frequently use the following two fact:
\paragraph{Fact 1:} $\frac{1}{m+1}\binom{n}{m} = \frac{1}{n+1}\binom{n+1}{m+1}$.
\paragraph{Fact 2:} $\frac{1}{(m+1)(m+2)}\binom{n}{m} = \frac{1}{(n+1)(n+2)}\binom{n+2}{m+2}$.
\begin{lemma} \label{lem: EW}
	Under the updating rule (\ref{eq: updatingrule}), there exists $\alpha_1\in [0,1], s.t.,  \forall j\in [n], \forall$ time $t$,
	\begin{align*}
		\mathbb{E}_{t,P} \big[W_t^{(j)}\big] = \alpha_1 I_n + (1-\alpha_1)A_n.
	\end{align*}
\begin{proof}
	Because of symmetry, we will fix $j$, say, $j=1$. So for simplicity, we omit superscript ${(j)}$ for all quantities in this proof, the subscript $t$ for $W$, and the subscript $t,P$ for $\mathbb{E}$, because they do not affect the proof.
	
	First we proof: $\exists \alpha_1, s.t. $
	\begin{equation} \label{eq: EW}
		\mathbb{E} [W]= \alpha_1 I_n + (1-\alpha_1)A_n.
	\end{equation}
	
	Let us understand the meaning of the element of $W$. For the $(k,l)$th element $W_{kl}$. From $X_{t+1} = V_tW$, we know that, $W_{kl}$ represents the portion that $\bm{v}_t^{(l)}$ will be in $\bm{x}_{t+1}^{(k)}$ (the block number ${j}$ has been omitted, as stated before). For $\bm{v}_t^{(l)}$ going into $\bm{x}_{t+1}^{(k)}$, it should first sent from $k$, received by node $b_t$ (also omit ${j}$), averaged with other $j$th blocks by node $b_t$, and at last sent from $b_t$ to $l$. For all pairs $(k,l)$ satisfied $k\not= l$, the expectations of $W_{kl}$ are equivalent because of the symmetry (the same packet drop rate, and independency). For the same reason, the expectations of $W_{kl}$ are also equivalent for all pairs $(k,l)$ satisfied $k=l$. But for two situations that $k=l$ and $k\not=l$, the expectation need not to be equivalent. This is because when the sending end $l$ is also the receiving end $k$, node $l$ (or $k$) will always keep its own portion $\bm{v}_t^{(l)}$ if $l$ is also the node dealing with block $j$, which makes a slight different.

\end{proof}
\end{lemma}

\begin{lemma} \label{lem: EWW}
	Under the updating rule (\ref{eq: updatingrule}), there exists $\alpha_1\in [0,1], s.t.,  \forall j\in [n], \forall$ time $t$, 
	\begin{align*}
		\mathbb{E}_{t,P} \Big[W_t^{(j)}{W_t^{(j)}}^{\top}\Big] = \alpha_1 I_n + (1-\alpha_1)A_n.
	\end{align*}
	Moreover, $\alpha_1$ satisfies:
	\begin{equation*}
	 	\alpha_1 \leq  \frac{np + (1-p)^n + nT_1 + nT_2 - 1}{n-1},
	\end{equation*}
	where
	\begin{align*}
		T_1 &= \frac{2\big(1-p^{n+1}-(n+1)(1-p)p^n - (n+1)n(1-p)^2p^{n-1}/2 - (1-p)^{n+1}\big)}{n(n+1)(1-p)^2},\\
		T_2 &=\frac{1-p^n-n(1-p)p^{n-1}-(1-p)^n}{(n-1)(1-p)}.
	\end{align*}
\begin{proof}
Similar to Lemma (\ref{lem: EW}), we fix $j=1$, and omit superscript $(j)$ for all quantities in this proof, the subscript $t$ for W and the subscript $t,P$ for $\mathbb{E}$.

Also similar to Lemma (\ref{lem: EW}), there exists $\alpha, s.t.$
\begin{equation} \label{eq: EWW}
	\mathbb{E} \big[WW^{\top}\big] = \alpha_1I_n + (1-\alpha_1)A_n
\end{equation}

The only thing left is to bound $\alpha_1$. From (\ref{eq: EWW}), we know that $\alpha_1 = \frac{n\mathbb{E}[WW^{\top}]_{(1,1)}-1}{n-1}$. After simple compute, we have $\mathbb{E}[WW^{\top}]_{(1,1)} = \mathbb{E} \Big[\sum\limits_{i=1}^n W_{1,i}^2\Big]$. So we have the following:
\begin{align*}
	\alpha_1 = \frac{n\mathbb{E} \Big[\sum\limits_{i=1}^n W_{1,i}^2\Big]-1}{n-1}.
\end{align*}
 Therefore, bounding $\alpha_1$ equals bounding $\mathbb{E} \Big[\sum\limits_{i=1}^n W_{1,i}^2\Big]$. Similar to Lemma (\ref{lem: EW}), we denote the event "node 1 deal with the first block" by $A$.

 \paragraph{Case 1: node 1 deal with the first block}  In this case, let's understand $W$ again. node 1 average the 1st blocks it has received, then broadcast to all nodes. Therefore, for every node $i$ who received this averaged block, $\bm{x}_{t+1}^{(i)}$ has the same value, in other words, the column $i$ of $W$ equals, or, $W_{1,i}$ equals to $W_{1,1}$.  On the other hand, for every node $i$ who did not receive this averaged block, they keep their origin model $v_t^{(i)}$. But $i\not=1$ (because node 1 deal with this block, itself must receive its own block), which means $W_{1,i} = 0$.
 
 Therefore, for $i\not=1, i\in [n]$, if node $i$ receive the averaged model, $W_{1,i} = W_{1,1}$. Otherwise, $W_{1,i} = 0$. Based on this fact, we can define the random variable $B_i$ for $i\not=1, i\in [n]$. $B_i = 1$ if node $i$ receive the averaged block., $B_i=0$ if node $i$ does not receive the averaged block. Immediately, we can obtain the following equation:
 \begin{equation} \label{eq: 21}
 	\mathbb{E} \Big[\sum\limits_{i=1}^n W_{1,i}^2\mid A\Big] = \mathbb{E}\Big[W_{1,1}^2\cdot\Big(\sum\limits_{i=2}^n B_i+1\Big)\mid A\Big] = \mathbb{E}[W_{1,1}^2 \mid A]\cdot\big(1+(n-1)\mathbb{E}[B_n\mid A]\big).
\end{equation}
The last equation results from that $A,B_2,\cdots, B_n$ are independent and $B_2, \cdots, B_n$ are from identical distribution.

First let's compute $\mathbb{E}[W_{1,1}^2 \mid A]$. If node $i$ received the 1st block from $m (0\le m \le n-1)$ nodes (except itself), then $W_{1,1} = 1/(m+1)$. The probability of this event is $\binom{n-1}{m}(1-p)^mp^{n-1-m}$. So we can obtain:
\begin{align*}
	\mathbb{E}[W_{1,1}^2 \mid A] &= \sum\limits_{m=0}^{n-1}\frac{1}{(m+1)^2}\binom{n-1}{m}(1-p)^mp^{n-1-m}\\
	&= \frac{1}{n^2}(1-p)^{n-1} + p^{n-1}+ \sum\limits_{m=1}^{n-2}\frac{1}{(m+1)^2}\binom{n-1}{m}(1-p)^mp^{n-1-m}\\
	&\le \frac{1}{n^2}(1-p)^{n-1} + p^{n-1}+\sum\limits_{m=1}^{n-2}\frac{2}{(m+1)(m+2)}\binom{n-1}{m}(1-p)^mp^{n-1-m}\\
	&\xlongequal{\text{Fact 2}}\frac{1}{n^2}(1-p)^{n-1}  + p^{n-1}+\sum\limits_{m=1}^{n-2}\frac{2}{n(n+1)}\binom{n+1}{m+2}(1-p)^mp^{n-1-m}\\
	&=\frac{1}{n^2}(1-p)^{n-1} + p^{n-1}+\frac{2}{n(n+1)}\sum\limits_{m=3}^{n}\binom{n+1}{m}(1-p)^{m-2}p^{n+1-m}\\
	&=\frac{1}{n^2}(1-p)^{n-1} + p^{n-1}
	  +\frac{2}{n(n+1)(1-p)^2}\sum\limits_{m=3}^{n}\binom{n+1}{m}(1-p)^{m}p^{n+1-m}\\
	&=\frac{1}{n^2}(1-p)^{n-1} +p^{n-1}\\
	& \quad + \frac{2\big(1-p^{n+1}-(n+1)(1-p)p^n - (n+1)n(1-p)^2p^{n-1}/2 - (1-p)^{n+1}\big)}{n(n+1)(1-p)^2}\\
	&\le \frac{1}{n^2}(1-p)^{n-1} + p^{n-1} +T_1,
\end{align*}
where we denote $T_1 := \frac{2\big(1-p^{n+1}-(n+1)(1-p)p^n - (n+1)n(1-p)^2p^{n-1}/2 - (1-p)^{n+1}\big)}{n(n+1)(1-p)^2}$.

Next let's compute $\mathbb{E}[B_n\mid A]$. $B_n$ is just a $0-1$ distribution, with success probability $1-p$. Therefore, $\mathbb{E}[B_n\mid A] = 1-p$.

Applying all these equations into (\ref{eq: 21}), we can get:
\begin{align*}
	\mathbb{E} \Big[\sum\limits_{i=1}^n W_{1,i}^2\mid A\Big] \leq & \left(\frac{1}{n^2}(1-p)^{n-1} + p^{n-1} + T_1\right)(1+(n-1)(1-p))\\
	\leq&\frac{(1-p)^{n-1}}{n} + p^{n-1} + nT_1\numberthis\label{case1_1}
\end{align*}

\paragraph{Case 2: node 1 does not deal with the first block and node 1 does not receive averaged block} We define a new event $C$, representing that node 1 does not receive the averaged block. So, \textbf{Case 2} equals the event $\bar{A}\cap C$. In this case, node 1 keeps its origin block $\bm{v}_t^{(1)}$, which means $W_{1,1} = 1$. 

Again due to symmetry, we can suppose that node $n$ deal with the first block. Then we can use the method in \textbf{Case 1}. But in this case, we only use $B_2, \cdots, B_{n-1}$, because node $n$ must receive its own block and node 1 does not receive averaged block, and we use $W_{1,n}$ instead of $W_{1,1}$. Then, we obtain:
\begin{align} \label{eq: EWW2}
	\mathbb{E} \Big[\sum\limits_{i=1}^n W_{1,i}^2\mid \bar{A}, C\Big] &= 1 + \mathbb{E}\Big[W_{1,n}^2\cdot\Big(\sum\limits_{i=2}^{n-1} B_i+1\Big)\mid \bar{A},C\Big]\\
	& = 1 + \mathbb{E}[W_{1,n}^2 \mid \bar{A},C]\cdot\big(1+(n-2)\mathbb{E}[B_{n-1}\mid \bar{A},C]\big).
\end{align}
Here, we similarly have $\mathbb{E}[B_{n-1}\mid \bar{A},C] = 1-p$, but we need to compute $\mathbb{E} [W_{1,n}^2\mid \bar{A},C]$. When the 1st block from node 1 is not received by node $n$, $W_{1,n} = 0$. If node 1's block is received, together with other $m, (0\le m\le n-2)$ nodes' blocks, then $W_{1,n} = 1/(m+2)$ (node $n$'s block is always received by itself). The probability of this event is $\binom{n-2}{m}(1-p)^{m+1}p^{n-2-m}$. Therefore,
\begin{align*}
	\mathbb{E}[W_{1,n}^2 \mid \bar{A},C] &= \sum\limits_{m=0}^{n-2}\frac{1}{(m+2)^2}\binom{n-2}{m}(1-p)^{m+1}p^{n-2-m}\\
	&=\frac{1}{n^2}(1-p)^{n-1} + \sum\limits_{m=0}^{n-3}\frac{1}{(m+2)^2}\binom{n-2}{m}(1-p)^{m+1}p^{n-2-m}\\
	&\le \frac{1}{n^2}(1-p)^{n-1} + \sum\limits_{m=0}^{n-3}\frac{1}{(m+2)(m+1)}\binom{n-2}{m}(1-p)^{m+1}p^{n-2-m}\\
	&\xlongequal{\text{Fact 2}}\frac{1}{n^2}(1-p)^{n-1} + \sum\limits_{m=0}^{n-3}\frac{1}{n(n-1)}\binom{n}{m+2}(1-p)^{m+1}p^{n-2-m}\\
	&=\frac{1}{n^2}(1-p)^{n-1} + \frac{1}{n(n-1)}\sum\limits_{m=2}^{n-1} \binom{n}{m}(1-p)^{m-1}p^{n-m}\\
	&=\frac{1}{n^2}(1-p)^{n-1} + \frac{1}{n(n-1)(1-p)}\sum\limits_{m=2}^{n-1} \binom{n}{m}(1-p)^{m}p^{n-m}\\
	&=\frac{1}{n^2}(1-p)^{n-1} + \frac{1-p^n-n(1-p)p^{n-1}-(1-p)^n}{n(n-1)(1-p)}\\
	&\le \frac{1}{n^2}(1-p)^{n-1} + \frac{1}{n}T_2,
\end{align*}
where $T_2:=\frac{1-p^n-n(1-p)p^{n-1}-(1-p)^n}{(n-1)(1-p)}$.

Applying these equations into (\ref{eq: EWW2}), we get:
\begin{align*}
	\mathbb{E} \Big[\sum\limits_{i=1}^n W_{1,i}^2\mid \bar{A}, C\Big] \leq & 1 + \left(\frac{1}{n^2}(1-p)^{n-1} + \frac{1}{n}T_2\right)\cdot(1+(n-2)(1-p))\\
	\leq & 1 + \frac{(1-p)^{n-1}}{n} +  T_2\numberthis\label{case2_1}
\end{align*}

\paragraph{Case 3: node 1 does not deal with the first block and node 1 receives averaged block} This is the event $\bar{C}\cap \bar{A}$. Similar to the analysis above, we have:
\begin{align*}
	\mathbb{E} \Big[\sum\limits_{i=1}^n W_{1,i}^2\mid \bar{A}, \bar{C}\Big] &= \mathbb{E}\Big[W_{1,1}^2\cdot\Big(\sum\limits_{i=2}^{n-1} B_i+2\Big)\mid \bar{A},\bar{C}\Big]\\
	&=\mathbb{E}\Big[W_{1,1}^2\mid \bar{A},\bar{C}\Big]\cdot (2+(n-2)\mathbb{E} [B_2 \mid \bar{A},\bar{C}])
\end{align*}

Similarly, we have $\mathbb{E} [B_2\mid \bar{A}, \bar{C}] = 1-p$. For $\mathbb{E}\Big[W_{1,1}^2\mid \bar{A},\bar{C}\Big]$, the argument is the same as $\mathbb{E} [W_{1,n}^2\mid \bar{A},C]$ in \textbf{Case 2}. So, we have:
\begin{align*}
	\mathbb{E}\Big[W_{1,1}^2\mid \bar{A},\bar{C}\Big] \le \frac{1}{n^2}(1-p)^{n-1} + \frac{1}{n}T_2.
\end{align*}

Applying these together, we can obtain:
\begin{align*}
	\mathbb{E} \Big[\sum\limits_{i=1}^n W_{1,i}^2\mid \bar{A}, \bar{C}\Big] \leq & \left(\frac{1}{n^2}(1-p)^{n-1} + \frac{1}{n}T_2\right)\cdot(2+(n-2)(1-p))\\
	\leq & \frac{(1-p)^{n-1}}{n} + T_2\numberthis\label{case3_1}
\end{align*}

Combined with three cases, $P(A) = 1/n$, $P(\bar{A},C) = p(n-1)/n$, and $P(\bar{A}, \bar{C}) = (1-p)(n-1)/n$, we have 
\begin{align*}
	\mathbb{E} \Big[\sum\limits_{i=1}^n W_{1,i}^2\Big] \leq& \frac{1}{n}\mathbb{E} \Big[\sum\limits_{i=1}^n W_{1,i}^2\mid A\Big] + \frac{p(n-1)}{n} \mathbb{E} \Big[\sum\limits_{i=1}^n W_{1,i}^2\mid \bar{A}, C\Big] 
	 + \frac{(1-p)(n-1)}{n}\mathbb{E} \Big[\sum\limits_{i=1}^n W_{1,i}^2\mid \bar{A}, \bar{C}\Big].	
\end{align*}
Combing the inequality above and \eqref{case1_1} \eqref{case2_1} \eqref{case3_1} together, we get
\begin{align*}
\mathbb{E} \Big[\sum\limits_{i=1}^n W_{1,i}^2\Big] \leq& \frac{(1-p)^n}{n^2} + \frac{p^n}{n} + T_1  + \frac{p(n-1)}{n} + \frac{p(1-p)^{n-1}}{n} + T_2p + \frac{(1-p)^n(n-1)}{n^2} + T_2(1-p) \\
\leq & p + \frac{(1-p)^n}{n} + T_1 + T_2  \\
\alpha_1 \leq & \frac{np + (1-p)^n + nT_1 + nT_2 - 1}{n-1}
\end{align*}

\end{proof}
\end{lemma}

\begin{lemma} \label{lem: EWAW}
	Under the updating rule (\ref{eq: updatingrule}), there exists $\alpha_2\in [0,1], s.t.,  \forall j\in [n], \forall$ time $t$,
	\begin{align*}
		\mathbb{E}_{t,P} \Big[W_t^{(j)}A_n{W_t^{(j)}}^{\top}\Big] = \alpha_2 I_n + (1-\alpha_2)A_n.
	\end{align*}
	Moreover, $\alpha_2$ satisfies: 
	\begin{equation*}
		\alpha_2 \le \frac{p(1+2T_3) + (1-p)^{n-1}}{n} + \frac{2p(1-p)^n}{n} + \frac{p^n(1-p)}{n^2} + T_1 + T_2,
	\end{equation*}
	where
	\begin{align*}
		T_1 &= \frac{2\big(1-p^{n+1}-(n+1)(1-p)p^n - (n+1)n(1-p)^2p^{n-1}/2 - (1-p)^{n+1}\big)}{n(n+1)(1-p)^2},\\
		T_2 &=\frac{1-p^n-n(1-p)p^{n-1}-(1-p)^n}{(n-1)(1-p)},\\
		T_3 &= \frac{n}{n-1}\left(1-p^{n-1} - (1-p)^{n-1}\right) + (1-p)^{n-1}.				
	\end{align*}
\begin{proof}
Similar to Lemma (\ref{lem: EW}) and Lemma (\ref{lem: EWW}), we fix $j=1$, and omit superscript $(j)$ for all quantities in this proof, the subscript $t$ for W and the subscript $t,P$ for $\mathbb{E}$. And we still use $A$ to denote the event "node 1 deal with the first block", use the binary random variable $B_i$ to denote whether node $i$ receive the averaged block. The definitions is the same to them in Lemma \ref{lem: EWW}.  

Again similar to Lemma (\ref{lem: EW}), there exists $\alpha_2, s.t.$
\begin{equation} \label{eq: EWAW}
	\mathbb{E} \big[WA_nW^{\top}\big] = \alpha_2I_n + (1-\alpha_2)A_n.
\end{equation}
The only thing left is to bound $\alpha_2$. From (\ref{eq: EWAW}), we know that $\alpha_2 = \frac{n\mathbb{E}[WA_nW^{\top}]_{(1,1)}-1}{n-1}$. After simple compute, we have $\mathbb{E}[WA_nW^{\top}]_{(1,1)} =  \mathbb{E} \left[\left(\sum\limits_{i=1}^n W_{1,i}\right)^2\right]/n$. So we have the following:
\begin{align*}
	\alpha_2 = \frac{\mathbb{E} \left[\left(\sum\limits_{i=1}^n W_{1,i}\right)^2\right]-1}{n-1}.
\end{align*}
Therefore, bounding $\alpha_2$ equals bounding $\mathbb{E} \left[\left(\sum\limits_{i=1}^n W_{1,i}\right)^2\right]$.

\paragraph{Case 1: node 1 deal with the first block} 
 In this case, $\sum\limits_{i=1}^n W_{1,i} = W_{1,1}\cdot \left(1+\sum\limits_{i=2}^n B_i \right)$, which means, $\left(\sum\limits_{i=1}^n W_{1,i}\right)^2 = W_{1,1}^2\cdot \left(1+\sum\limits_{i=2}^n B_i\right)^2$. Similar to Lemma \ref{lem: EWW}, $A$ and $\left\{B_i\right\}_{i=2}^n$ are independent, so we have:
 \begin{equation*}
 	\mathbb{E} \left[\left(\sum\limits_{i=1}^n W_{1,i}\right)^2\mid A\right] = \mathbb{E} \left[W_{1,1}^2\mid A\right]\cdot \mathbb{E} \left[ \left(1+\sum\limits_{i=2}^n B_i\right)^2\mid A\right].
\end{equation*}

From Lemma \ref{lem: EWW}, we have 
\begin{align*}
\mathbb{E} \left[W_{1,1}^2\mid A\right] \le \frac{1}{n^2}(1-p)^{n-1} + p^{n-1} +T_1.
\end{align*}
For $\mathbb{E} [(1+\sum\limits_{i=2}^n B_i)^2\mid A]$, since $\left\{B_i\right\}_{i=2}^n$ are independent, we have the following:
\begin{align*}
	\mathbb{E} \left[ \left(1+\sum\limits_{i=2}^n B_i\right)^2\mid A\right]&=\left(\mathbb{E} \left[1+\sum\limits_{i=2}^n B_i\mid A\right]\right)^2 + \mathrm{Var}\left[1+\sum\limits_{i=2}^n B_i\mid A\right]\\
	& = \left(1+\left(n-1\right)\left(1-p\right)\right)^2 + (n-1)p(1-p)
\end{align*}

Combined these together, we obtain:
\begin{align*}
	&\mathbb{E} \left[\left(\sum\limits_{i=1}^n W_{1,i}\right)^2\mid A\right]\\
	\leq & \left(\frac{1}{n^2}(1-p)^{n-1} + p^{n-1} +T_1\right)\left(\left(1+\left(n-1\right)\left(1-p\right)\right)^2 + (n-1)p(1-p)\right)\\
	\leq & \frac{(1-p)^{n-1}(1+(n-1)(1-p))^2}{n^2} + p^{n-1}(1+(n-1)(1-p))^2 + \frac{p^n(1-p) + (1-p)^np}{n} + n^2 T_1
\end{align*}

\paragraph{Case 2: node 1 does not deal with the first block and node 1 does not receive averaged block} In this case, $\sum\limits_{i=1}^n W_{1,i} = 1 + W_{1,n}\cdot\left(\sum\limits_{i=2}^{n-1}B_i + 1\right)$ (suppose node $n$ deal with the first block). So we have:
\begin{equation*}
	\left(\sum\limits_{i=1}^n W_{1,i}\right)^2 = 1 + 2W_{1,n}\left(\sum\limits_{i=2}^{n-1}B_i + 1\right) + W_{1,n}^2\left(\sum\limits_{i=2}^{n-1}B_i+1\right)^2,
\end{equation*}
which means (notice $W_{1,n}$ and $\{B_i\}_{i=2}^{n-1}$ are independent),
\begin{equation*}
	\mathbb{E} \left[\left(\sum\limits_{i=1}^n W_{1,i}\right)^2\mid \bar{A},C\right] = 1 + 2\mathbb{E}\left[W_{1,n}\mid \bar{A},C\right]\mathbb{E} \left[\sum\limits_{i=2}^{n-1}B_i+1\mid \bar{A},C\right] + \mathbb{E}\left[W_{1,n}^2\mid \bar{A},C\right]\mathbb{E} \left[\left(\sum\limits_{i=2}^{n-1} B_i + 1\right)^2\mid \bar{A},C\right].
\end{equation*}

First let's consider $\mathbb{E}[W_{1,n}\mid \bar{A},C]$. Similar to the analysis of \textbf{Case 2} in Lemma \ref{lem: EWW} except instead first moment of second moment, we have:
\begin{align*}
	\mathbb{E}[W_{1,n}\mid \bar{A},C] &= \sum\limits_{m=0}^{n-2} \frac{1}{m+2}\binom{n-2}{m}(1-p)^{m+1}p^{n-2-m}\\
	&= \frac{1}{n}(1-p)^{n-1} + \sum\limits_{m=0}^{n-3} \frac{1}{m+2}\binom{n-2}{m}(1-p)^{m+1}p^{n-2-m}\\
	&\le \frac{1}{n}(1-p)^{n-1} + \sum\limits_{m=0}^{n-3} \frac{1}{m+1}\binom{n-2}{m}(1-p)^{m+1}p^{n-2-m}\\
	&\xlongequal{\text{Fact 1}}\frac{1}{n}(1-p)^{n-1} + \sum\limits_{m=0}^{n-3} \frac{1}{n-1}\binom{n-1}{m+1}(1-p)^{m+1}p^{n-2-m}\\
	&=\frac{1}{n}(1-p)^{n-1} + \frac{1}{n-1}\sum\limits_{m=1}^{n-2}\binom{n-1}{m}(1-p)^mp^{n-1-m}\\
	&=\frac{1}{n}(1-p)^{n-1} + \frac{1}{n-1}\left(1-p^{n-1} - (1-p)^{n-1}\right)\\
	& = \frac{T_3}{n},
\end{align*}
where we denote $T_3:= \frac{n}{n-1}\left(1-p^{n-1} - (1-p)^{n-1}\right) + (1-p)^{n-1}$.

Next, from Lemma \ref{lem: EWW}, we have 
\begin{align*}
\mathbb{E}\left[W_{1,n}^2\mid \bar{A},C\right]\le \frac{1}{n^2}(1-p)^{n-1} + \frac{1}{n}T_2.
\end{align*}

Next we deal with item with $B_i$.
We have the following:
\begin{equation*}
	\mathbb{E} \left[\sum\limits_{i=2}^{n-1}B_i + 1\mid \bar{A},C\right] = (n-2)(1-p) + 1
\end{equation*}
\begin{align*}
	\mathbb{E} \left[\left(\sum\limits_{i=2}^{n-1}B_i+1\right)^2\mid \bar{A},C\right] &= 1 + 2(n-2)\mathbb{E} [B_2 \mid \bar{A},C] + \mathbb{E} \left[\left(\sum\limits_{i=2}^{n-1}B_i\right)^2\mid \bar{A},C\right]\\
	&=1+2(n-2)(1-p) + \left(\mathbb{E} \left[\sum\limits_{i=2}^{n-1}B_i\mid \bar{A},C\right]\right)^2 + \mathrm{Var}\left[\sum\limits_{i=2}^{n-1}B_i\mid \bar{A},C\right]\\
	&=1+2(n-2)(1-p) + (n-2)^2(1-p)^2 + (n-2)p(1-p)
\end{align*}
Combining those terms together we get
\begin{align*}
&\mathbb{E} \left[\left(\sum\limits_{i=1}^n W_{1,i}\right)^2\mid \bar{A},C\right]\\
\leq & 
1 + \frac{(1-p)^{n-1}}{n^2} + \frac{(n-2)(1-p)^n}{n} + \frac{p(1-p)^n}{n} + nT_2  + 2T_3
\end{align*}

\paragraph{Case 3: node 1 does not deal with the first block and node 1 receives averaged block} In this case, $\sum\limits_{i=1}^n W_{1,i} = W_{1,1}\cdot \left(\sum\limits_{i=2}^{n-1}B_i + 2\right)$. So we have: 
\begin{equation*}
	\left(\sum\limits_{i=1}^n W_{1,i}\right)^2 = W_{1,1}^2\cdot \left(\sum\limits_{i=2}^{n-1}B_i + 2\right)^2,
\end{equation*}
which means (notice that $W_{1,1}$ and $\{B_i\}_{i=2}^{n-1}$ are independent)
\begin{equation*}
	\mathbb{E} \left[\left(\sum\limits_{i=1}^n W_{1,i}\right)^2 \mid \bar{A}, \bar{C}\right] = \mathbb{E} \left[W_{1,1}^2\mid \bar{A}, \bar{C}\right]\cdot \mathbb{E} \left[\left(\sum\limits_{i=2}^{n-1}B_i + 2\right)^2\mid \bar{A}, \bar{C}\right].
\end{equation*}

Similar to Lemma \ref{lem: EWW}, $\mathbb{E} \left[W_{1,1}^2\mid \bar{A}, \bar{C}\right]$ is the same as $\mathbb{E}\left[W_{1,n}^2\mid \bar{A},C\right]$.
\begin{align*}
\mathbb{E} \left[W_{1,1}^2\mid \bar{A}, \bar{C}\right] = \frac{1}{n^2}(1-p)^{n-1} + \frac{1}{n}T_2
\end{align*}

Also, we have the following:
\begin{align*}
	\mathbb{E} \left[\left(\sum\limits_{i=2}^{n-1}B_i + 2\right)^2\mid \bar{A}, \bar{C}\right] &= \left(\mathbb{E} \left[\sum\limits_{i=2}^{n-1}B_i + 2\mid \bar{A}, \bar{C}\right]\right)^2 + \mathrm{Var}\left[\sum\limits_{i=2}^{n-1}B_i+2\mid\bar{A},\bar{C}\right]\\
	&=[(n-2)(1-p)+2]^2 + (n-2)p(1-p)
\end{align*}

So we have
\begin{align*}
&\mathbb{E} \left[\left(\sum\limits_{i=1}^n W_{1,i}\right)^2 \mid \bar{A}, \bar{C}\right]\\
\leq & \frac{(1-p)^{n-1}((n-2)(1-p) + 2)^2}{n^2} + \frac{p(1-p)^n}{n} + nT_2
\end{align*}

Combining these inequalities together, we have the following:
\begin{align*}
	&\mathbb{E} \left[\left(\sum\limits_{i=1}^n W_{1,i}\right)^2\right]\\
	=&\mathbb{E} \left[\left(\sum\limits_{i=1}^n W_{1,i}\right)^2\mid A\right]P(A) + \mathbb{E} \left[\left(\sum\limits_{i=1}^n W_{1,i}\right)^2\mid \bar{A},C\right]P(\bar{A},C)\\
	& + \mathbb{E} \left[\left(\sum\limits_{i=1}^n W_{1,i}\right)^2 \mid \bar{A}, \bar{C}\right]P(\bar{A},\bar{C})\\
	\leq & \frac{(1-p)^{n-1}(1+(n-1)(1-p))^2}{n^3} + \frac{p^{n-1}(1+(n-1)(1-p))^2}{n} + \frac{2p^n(1-p) + (1-p)^np}{n^2}\\
	& + n T_1 +   \frac{(n-1)p}{n}\left( 1 + \frac{(1-p)^{n-1}}{n^2} + \frac{(n-2)(1-p)^n}{n} + \frac{p(1-p)^n}{n} + nT_2  + 2T_3 \right) \\
	& + \frac{(n-1)(1-p)}{n} \left( \frac{(1-p)^{n-1}((n-2)(1-p) + 2)^2}{n^2} + \frac{p(1-p)^n}{n} + nT_2 \right)\\
	\leq & \frac{(1-p)^{n-1}}{n} + \frac{p^{n-1}(1+(n-1)(1-p))^2}{n} + \frac{p(1-p)^n+p^n(1-p)}{n^2} + nT_1\\
	& + p + \frac{p(1-p)^n}{n} + nT_2 + 2T_3p + \frac{(n-1)(1-p)^n}{n}\\
	\leq & p(1+2T_3) + (1-p)^{n-1} + \frac{2p(1-p)^n}{n} + \frac{p^n(1-p)}{n^2} + n(T_1+T_2),
\end{align*}
and
\begin{align*}
\alpha_2 \leq & \frac{p(1+2T_3) + (1-p)^{n-1}}{n} + \frac{2p(1-p)^n}{n} + \frac{p^n(1-p)}{n^2} + T_1 + T_2.
\end{align*}

\end{proof}
\end{lemma}